\documentclass[5p,authoryear]{elsarticle}
\usepackage[fleqn]{amsmath}
\usepackage{amstext}
\usepackage{amsfonts}
\usepackage{amssymb}
\usepackage{amscd}
\usepackage{graphicx}
\usepackage{placeins}
\usepackage{enumerate}
\usepackage{float}
\usepackage{mathtools}
\usepackage{xcolor}
\usepackage{multicol}
\usepackage{hyperref}
\usepackage{adjustbox}
\usepackage{appendix}
\usepackage{lineno}
\usepackage{ogonek}
\usepackage[percent]{overpic}

\journal{Mathematical Social Sciences}

\newtheorem{theorem}{Theorem}

\newtheorem{conjecture}[theorem]{Conjecture}
\newtheorem{corollary}[theorem]{Corollary}

\newtheorem{definition}[theorem]{Definition}

\newtheorem{notation}[theorem]{Notation}

\newtheorem{proposition}[theorem]{Proposition}
\newtheorem{remark}[theorem]{Remark}

\newenvironment{proof}[1][Proof]{\noindent\textbf{#1.} }{\ \rule{0.5em}{0.5em}}
\input{tcilatex}

\definecolor{color1}{rgb}{0.368417, 0.506779, 0.709798}
\definecolor{color2}{rgb}{0.880722, 0.611041, 0.142051}
\definecolor{color3}{rgb}{0.560181, 0.691569, 0.194885}
\definecolor{color4}{rgb}{0.922526, 0.385626, 0.209179}
\definecolor{color5}{rgb}{0.528488, 0.470624, 0.701351}

\begin{document}

\begin{frontmatter}
	
\title{Average Weights and Power\\
in Weighted Voting Games\tnoteref{tn1}}


\author[1]{Daria Boratyn\fnref{fn1}\fnref{fn2}}
\ead{daria.boratyn@im.uj.edu.pl}

\author[2]{Werner Kirsch\fnref{fn1}\fnref{fn3}}
\ead{werner.kirsch@fernuni-hagen.de}

\author[3]{Wojciech S\l omczy\'{n}ski\fnref{fn1}\fnref{fn2}}
\ead{wojciech.slomczynski@im.uj.edu.pl}

\author[4]{Dariusz~Stolicki\fnref{fn1}\fnref{fn4}}
\ead{dariusz.stolicki@uj.edu.pl}

\author[5]{Karol \.{Z}yczkowski\fnref{fn1}\fnref{fn5}}
\ead{karol.zyczkowski@uj.edu.pl}

\fntext[fn1]{Jagiellonian Center for Quantitative Research in Political Science, Jagiellonian University, ul. Wenecja 2, 31-117 Krak\'{o}w, Poland.}
\fntext[fn2]{Institute of Mathematics, Jagiellonian University.}
\fntext[fn3]{Fakult\"{a}t f\"{u}r Mathematik und Informatik, FernUniversit\"{a}t Hagen, Germany; Dimitris-Tsatos-Institut f\"ur Europ\"aische Verfassungswissenschaften.}
\fntext[fn4]{Institute of Political Science and International Relations, Jagiellonian University.}
\fntext[fn5]{Institute of Physics, Jagiellonian University.}

\tnotetext[tn1]{Dedicated to the memory of Fritz Haake (1941--2019) -- a theoretical physicist who promoted German-Polish scientific collaboration.}

\begin{abstract}
We investigate a class of weighted voting games for which weights are
randomly distributed over the standard probability simplex. We provide close-formed formulae
for the expectation and density of the distribution of weight of the $k$-th
largest player under the uniform distribution. We analyze the average voting
power of the $k$-th largest player and its dependence on the quota,
obtaining analytical and numerical results for small values of $n$ and a
general theorem about the functional form of the relation between the
average Penrose--Banzhaf power index and the quota for the uniform measure
on the simplex. We also analyze the power of a collectivity to act (Coleman
efficiency index) of random weighted voting games, obtaining analytical upper bounds therefor.

\end{abstract}


\end{frontmatter}

\section{Introduction}

An $n$--player \emph{weighted voting game} $G$ is described by a \emph{%
weight vector} $\mathbf{w}:=\left( w_{1},\dots ,w_{n}\right) \in
\Delta _{n}$, where $\Delta _{n}$ is the standard $\left( n-1\right) $%
--dimensional probability simplex, and a \emph{qualified majority quota }$%
q\in \mathbb{(}\frac{1}{2},1\mathbb{]}$. In such a game, the set of winning
coalitions $\mathcal{W\subset P}(V)$, where $V$ is the set of players, is
defined as follows:%
\begin{equation}
\mathcal{W}:=\bigg\{ Q\subset V:\sum_{v\in Q}w_{v}\geq q\bigg\} .
\end{equation}%
We denote the set of all $n$-player weighted voting games by $\mathcal{G}%
_{n} $.

By \emph{random weighted voting game} we mean a weighted voting game
in which the number of players $n$ and the quota $q$ are fixed, and the
weight vector $\mathbf{w}$ is drawn from the standard probability simplex $%
\Delta _{n}$ with some probability measure. Such games seem to be
interesting for a number of reasons. First, the analysis of random weighted
voting games enhances our understanding of weighted voting games in general.
One of the major challenges in the field lies in the fact that generic
results are\ usually rather difficult to obtain, while the behavior of
weighted voting games in specific cases depends heavily on the
characteristics of the specific weight vector and is often subject to
number-theoretic peculiarities. For instance, some of the fundamental
questions touch the relationship between the quota $q$ and the influence of
individual players or efficiency of the system as a whole. Yet, for fixed
weight vectors those dependencies are not only discontinuous, but highly
erratic. Randomizing the weights, and thus averaging them over the simplex,
smooths out the peculiarities of specific weight vectors, revealing hitherto
unobserved regularities.

Second, randomizing the weights is likely to be of interest from the
standpoint of voting rule design. Rule design tends to take place before
players' weights are fixed, and thus any predictions regarding the effects
of the rules must, to the extent such effects depend on voting weights,
necessarily be probabilistic. Also, just like players' preferences are
treated as random to abstract away from particular issues and focus the
attention on the voting rules themselves \citep{Roth88a}, treating voting
weights as random further abstracts away the particular configuration of
players and brings other parameters (such as the number of players or the
quota) into the forefront.

Obviously, the characteristics of a random weighted voting game depend on
the choice of the probability measure. In the present article, we focus on
the uniform (Lebesgue) measure (which is equivalent to the familiar
Impartial Anonymous Culture Model used in computational social choice, see 
\citealp{KugaNagatani74,GehrleinFishburn76a}). For that measure we obtain exact
closed-form formulae for the expectation and density of the distribution of
voting weight of the $k$--th largest player, an analytical formula for the
expected values of product--moments of voting weights, a general theorem
about the functional form of the relation between the expected values of the
absolute and normalized Penrose--Banzhaf indices of the $k$--th largest
player and the quota, the characteristic function of the distribution of
coalition weights, and an approximation of the Coleman efficiency index (the
power of a collectivity to act). All of those results constitute an original
contribution of the paper. We further outline several applications of those
results in the field of mathematical voting theory and in some other areas.

\subsection{Related work\label{sec:relatedWork}}

The notion of \emph{voting power}, i.e., a player's influence on the outcome
of the game, which, as demonstrated by \citet{Penrose46}, is not
necessarily proportional to the player's weight, is of fundamental
importance to the study of voting systems. The two of the most popular
voting power indices have been introduced by \citet%
{ShapleyShubik54} and by \citet{Banzhaf64}. Both define the
voting power of a player $v$ in terms of the probability that their vote is
decisive, but differ in their definition of the probability measure on the
set of voting outcomes: the \emph{Shapley--Shubik index} treats each
permutation of players as equiprobable, while the \emph{Penrose--Banzhaf
index} assigns equal probabilities to all combinations of players. In
addition, there are two versions of the Penrose--Banzhaf index in common
use: one is defined as the probability of a player $v$ casting a decisive
vote and is known as the \emph{non-normalized} or \emph{absolute
Penrose--Banzhaf index}, $\psi _{v}$ \citep{DubeyShapley79}, while the other
one, $\beta _{v}$, is further normalized in order to ensure that the vector $%
\boldsymbol{\beta }:=\left( \beta _{1},\dots ,\beta _{n}\right) $ lies in
the probability simplex $\Delta _{n}$. Note that the vector of
Shapley--Shubik indices always lies in $\Delta _{n}$, hence there is no need
for further normalization.

It is well known that each player's voting power depends not only on the
weight vector, but also on the quota \citep%
{FelsenthalMachover98,LeechMachover03}. The relationship between the quota
and the Penrose--\linebreak Banzhaf power index for a fixed weight vector has been
investigated by \citet{Leech02a} and more generally by \citet{ZuckermanEtAl12}%
, with the latter reporting several results on, \textit{inter alia}, the
upper and lower bounds of the ratio and difference between a player's weight
and their normalized Penrose--Banzhaf index. Analytical results about the
values of the Penrose--Banzhaf index depending on the quota are available
primarily for extreme quotas: the \emph{Penrose limit theorem} \citep%
{Penrose46,Penrose52}, proven under certain technical assumptions by \citet%
{LindnerMachover04}, provides that for $q=1/2$ and all $i,j\in V$, the ratio 
$\psi _{i}/\psi _{j}$ converges to $w_{i}/w_{j}$ as $n\rightarrow \infty $.
On the other hand, it is easy to notice that as $q\rightarrow 1$, the values
of $\psi _{i}$ and $\beta _{i}$ converge to $2^{1-n}$ and $1/n$,
respectively, regardless of $\mathbf{w}$. \citet%
{SlomczynskiZyczkowski06,SlomczynskiZyczkowski07} have established
that $q^{\ast }:=\frac{1}{2}\left( 1+(\sum_{i=1}^{n}w_{i}^{2})^{-1}\right) $
is a good approximation of the quota minimizing the distance $\left\Vert 
\mathbf{w}-\boldsymbol{\beta }\right\Vert _{2}$. For the discussion of the
political significance of this quota, see \citet{Grimmett19}. Therefore, if $%
\mathbf{w}$ is uniformly distributed on $\Delta _{n}$, then $\mathbb{E}%
(q^{\ast })\approx \frac{1}{2}+\frac{1}{\sqrt{\pi n}}$ \citep%
{ZyczkowskiSlomczynski13}. Upper bounds for the deviation between weights
and Penrose--Banzhaf indices have been provided by \citet{Kurz18}. The
relationship between the number of dummy players, i.e., such players $v$
that $\beta _{v}=0$, and the quota has been studied by \citet%
{BarthelemyEtAl13}.

The case of random weights has been investigated only for the Shapley--Shubik
index $S_{v}$. The issue of selecting quotas maximizing and minimizing the
Shapley--Shubik power of a given player is analyzed by \citet{ZickEtAl11},
who note that testing whether a given quota does so is an NP-hard
problem. They also note that for a large range of quotas starting
with $1/2$, the Shapley--Shubik power of a small player tends to be stable
and close to their weight. \citet{JelnovTauman14} established that
if $\mathbf{w}$ is uniformly distributed on $\Delta _{n}$, the expected
ratio of Shapley--Shubik index to weight approaches $1$ as $n\rightarrow
\infty $. \citet{BachrachEtAl17} identify certain number-theoretic
artifacts in the relationship between $S_{v}$ and $q$ for weights drawn from
a multinomial distribution and normalized, and provides a lower bound for
the expected index $S_{v}$ of the smallest player $v$. A problem similar to ours
is posed by \citet{FilmusEtAl16}, who provide a closed-form
characterization of the Shapley values of the largest and smallest players
for $\mathbf{w}$ drawn from a uniform distribution on $\Delta _{n}$ or
obtained by normalizing $n$ independent random variables drawn from a
uniform distribution. Finally, \cite{BachrachEtAl16a} give a closed-form
formula for the Shapley--Shubik power index in games with
super-increasing weights.

Numerous works analyze weighted voting games in a variety of empirical
settings, including the Council of the European Union \citep%
{LaruelleWidgren98,Leech02a,FelsenthalMachover04,ZyczkowskiCichocki10,ZyczkowskiSlomczynski13},
the U.S. Electoral College \citep{Owen75,Miller13}, the International
Monetary Fund \citep{Leech02b,LeechLeech13}, the U.N. Security Council \citep%
{StrandRapkin11} and joint stock companies \citep{Leech02}. The list of
references is by no means complete, but demonstrates that the relevance of
the subject goes far beyond purely academic.

\section{Voting Weight of the $k$--th Largest Player}

\subsection{Introduction}

Let $\Delta _{n}$ be the standard $\left( n-1\right) $--dimensional
probability simplex, which represents the set of normalized weight vectors.
We consider a random weighted voting game, where the weight vector $\mathbf{W%
}\in \Delta _{n}$\ is a random variable with the uniform probability
distribution, which will be thereafter denoted as $\func{Unif}\left( \Delta
_{n}\right) $. Since the uniform measure is symmetric, the players are
indistinguishable \textit{a priori}. But note that the coordinates of $%
\mathbf{W}$, i.e., the voting weights of the players, can almost surely be
strictly ordered. This ordering provides a natural basis for distinguishing
the players \textit{a posteriori}.

\begin{notation}
For $k=1,\dots ,n$ we denote the $k$--th largest coordinate of a vector $%
\mathbf{x}\in \mathbb{R}^{n}$ as $x_{k}^{\downarrow }$.
\end{notation}

We start with the simplest question: what is the expected value and density
of the distribution of voting \linebreak weight of the $k$--th largest player in a
random weighted voting game? While the coordinates of $\mathbf{W}$\ can be
thought of as a sample of random variables, and $W_{k}^{\downarrow }$ as the 
$k$--th largest order statistic of that sample, virtually all results in the
field assume that order statistics are computed for a sample of independent
variables, which is manifestly not the case for the barycentric coordinates
of a vector drawn from a simplex. For that reason, the problem can be
considered non-trivial.

\subsection{Expected value: barycenter of the asymmetric simplex}

Each ordering of the coordinates of a generic weight vector $\mathbf{w}$\textbf{,} $%
w_{1}^{\downarrow }>w_{2}^{\downarrow }>\dots >w_{n}^{\downarrow }$,
corresponds to dividing the entire simplex $\Delta _{n}$ into $n!$
asymmetric parts and selecting one of them, which we will denote as $\tilde{%
\Delta}_{n}$. If $\mathbf{W}$ is drawn from the uniform distribution on $%
\Delta _{n}$, the \emph{ordered weight vector} $\mathbf{W}^{\downarrow
}=(W_{1}^{\downarrow },W_{2}^{\downarrow },\dots ,W_{n}^{\downarrow })$ is
uniformly distributed on the \emph{asymmetric simplex} $\tilde{\Delta}_{n}$ with vertices $(1,0,\dots ,0)$%
, $\frac{1}{2}(1,1,0,\dots ,0)$, \dots , $\frac{1}{n}(1,1,\dots ,1)$, see Fig. \ref{fig:simpl3}.

\begin{figure}[h]
	\centering
	$\vcenter{\hbox{\includegraphics[width=0.48\linewidth]{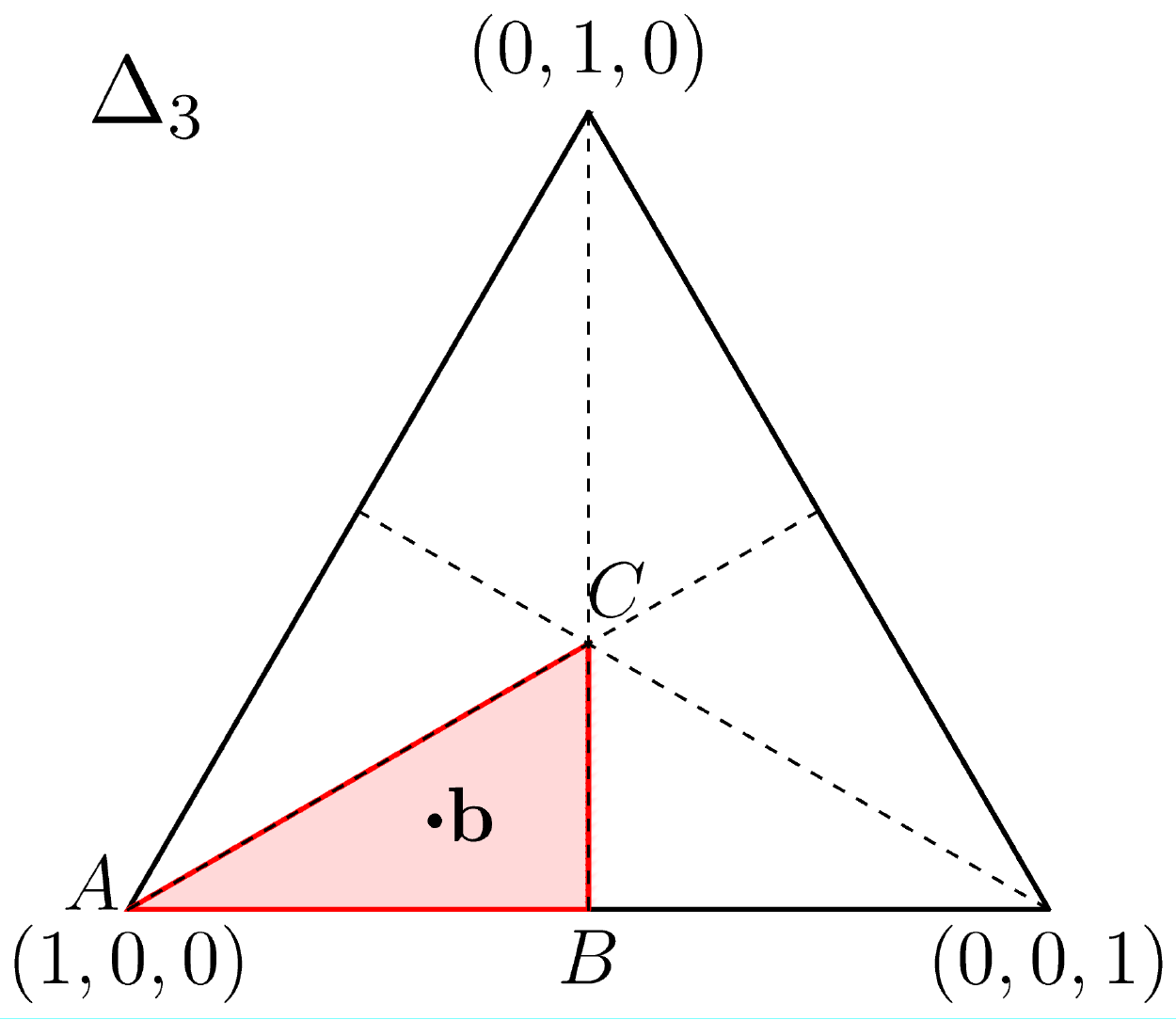}}}$
	$\vcenter{\hbox{\includegraphics[width=0.48\linewidth]{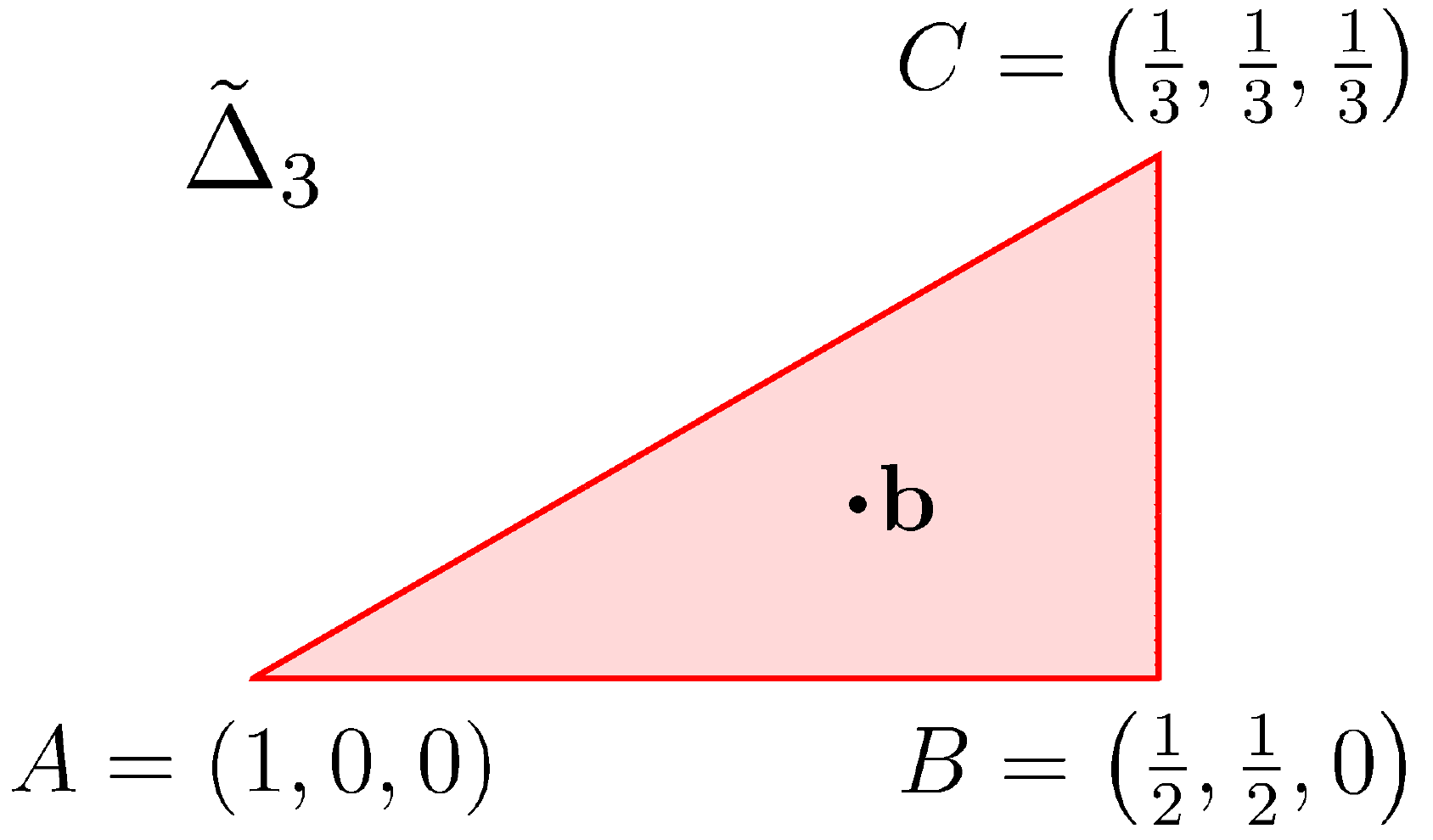}}}$
	\caption{The case of $n=3$:
		probability simplex $\Delta_3$ 
		as well as the asymmetric simplex ${\tilde \Delta_3}$ 
		with vertices $A,B,C$ and the barycenter $\mathbf{b}=(A+B+C)/3=(11,5,2)/18$.}
\label{fig:simpl3}
\end{figure}

The expected value of $\mathbf{W}^{\downarrow }$ coincides with the
barycenter $\mathbf{b}$ of $\tilde{\Delta}_{n}$. The $k$--th coordinate of
that barycenter, $b_{k}$, for $k=1,\dots ,n$, can be expressed by the sum of 
\emph{harmonic numbers} $H_{l}:=\sum_{j=1}^{l}1/j$, as follows:%
\begin{equation}
b_{k}=(H_{n}-H_{k-1})/n=\frac{1}{n}\sum_{j=k}^{n}\frac{1}{j}.
\end{equation}

Thus we obtain an explicit formula, valid for an arbitrary number of players 
$n$,\ for the expected voting weight of the $k$--th largest player:

\begin{proposition}
If $\mathbf{W\sim }\func{Unif}\left( \Delta _{n}\right) $, then for each\linebreak $k=1,\dots ,n$:
\begin{equation}
\mathbb{E}(W_{k}^{\downarrow })=b_{k}=\frac{1}{n}\sum_{j=k}^{n}\frac{1}{j}.
\end{equation}
\end{proposition}

E.g., for $n=3$ the expected ordered random probability vector is $\mathbb{E}%
(\mathbf{W}^{\downarrow })=(11,5,2)/18$, while for $n=6$ one obtains $%
\mathbb{E}(\mathbf{W}^{\downarrow })=(147,87,57,37,22,10)/360$. Note that
for a large $n$ the harmonic numbers scale as $\ln n+\gamma $, where $\gamma 
$ is the Euler--Mascheroni constant, so the first coordinate scales as $\ln
n/n$, the median coordinate as $\ln 2/n$, and the smallest coordinate as $1/n^{2}$.

\subsection{Densities}

More generally, we obtain the following theorem, with proof given in
the Appendix:

\begin{theorem}
\label{thm:kPlayer}If $\mathbf{W\sim }\func{Unif}\left( \Delta _{n}\right) $%
, then $W_{k}^{\downarrow }$, $k=1,\dots ,n$, is distributed according to an
absolutely continuous distribution supported on $[1/n,1]$ for $k=1$ and on $%
[0,1/k]$ for $k>1$, with piecewise polynomial density $f_{n,k}:[0,1]%
\rightarrow \mathbb{R}$ given by:%
\begin{multline}
f_{n,k}\left( x\right) :=n(n-1)\dbinom{n-1}{k-1}\times \\
\sum\limits_{j=k}^{\min
\left( n,\left\lfloor 1/x\right\rfloor \right) }(-1)^{j-k}\dbinom{n-k}{j-k}%
\left( 1-jx\right) ^{n-2}.  \label{densKthPlayer}
\end{multline}
\end{theorem}

\begin{figure}[h]
	\centering
	\includegraphics[width=\linewidth]{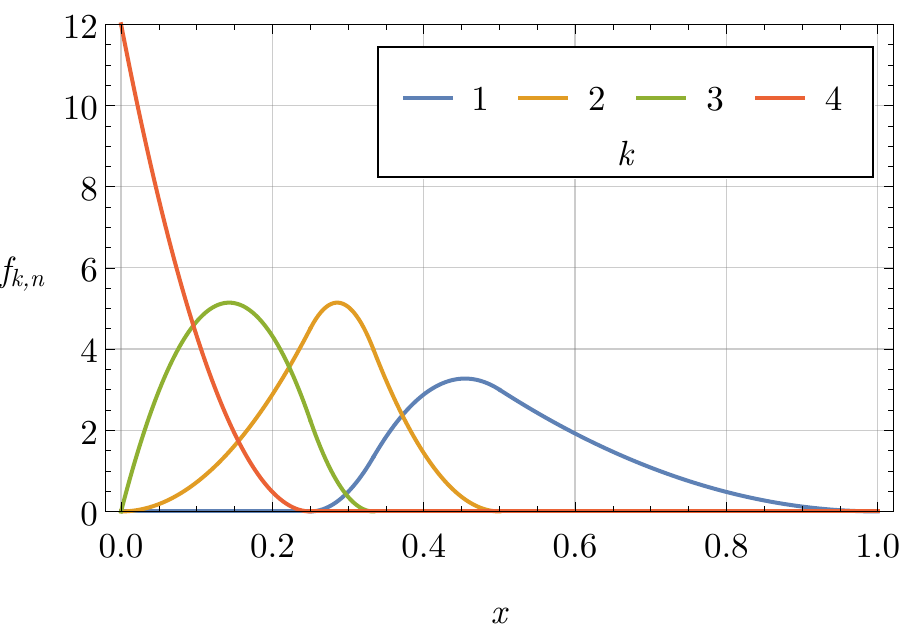}
	\caption{Densities of the distributions of the voting weight of the $k$-th largest of $4$ players for $k = 1, \dots, 4$.}
	\label{fig:kDensity}
\end{figure}%

\begin{remark}
The above result can also be obtained from results on the order statistics
of uniform spacings \citep{Darling53,RaoSobel80,Devroye81}. Nevertheless, we
believe that the approach described in the Appendix is more promising in the
context of a possible generalization to the family of Dirichlet
distributions.
\end{remark}

\begin{remark}
Elementary techniques of real analysis are sufficient to demonstrate that $%
f_{n,k}$ is smooth of class $\mathcal{C}^{n-3}$ for $n>2$.
\end{remark}

\begin{remark}
Theorem \ref{thm:kPlayer} extends an earlier result by \citet{QeadanEtAl12}%
, where, \textit{inter alia}, closed-form formulae are obtained for the
joint density of a sum and maximum of exponentially distributed i.i.d.
random variables. It suffices to note that the normalized vector of $n$
independent exponential random variables with mean $1$ is uniformly
distributed on $\Delta _{n}$ \citep{Jambunathan54}.
\end{remark}

\begin{proposition}
\label{propOrderStats}If $\mathbf{W\sim }\func{Unif}\left( \Delta
_{n}\right) $, then for $k=1,\dots ,n$,%
\begin{equation}
\left( W_{1}^{\downarrow },\dots ,W_{n}^{\downarrow }\right) \overset{d}{=}%
\bigg(\sum_{j=1}^{n}\frac{W_{j}}{j},\dots ,\sum_{j=n}^{n}\frac{W_{j}}{j}\bigg).  \label{ordSumWeights}
\end{equation}
\end{proposition}

\begin{proof}
By \citet{Jambunathan54} we can assume that for $j=1,\dots ,n$,%
\begin{equation}
W_{j}=\frac{X_{j}}{\sum_{i=1}^{n}X_{i}},  \label{jambunathan}
\end{equation}%
where $X_{1},\dots ,X_{n}\sim \func{Exp}\left( 1\right) $ are independent
random variables. Then by R\'{e}nyi representation formula \citep{Renyi53},%
\begin{equation}
\left( X_{1}^{\downarrow },\dots ,X_{n}^{\downarrow }\right) \overset{d}{=}%
\bigg(\sum_{j=1}^{n}\frac{X_{j}}{j},\dots ,\sum_{j=n}^{n}\frac{X_{j}}{j}\bigg).  \label{renyi}
\end{equation}%
Moreover,%
\begin{equation}
\sum_{k=1}^{n}X_{k}^{\downarrow
}=\sum_{k=1}^{n}X_{k}=\sum_{k=1}^{n}\sum_{j=k}^{n}\frac{X_{j}}{j}.
\label{renyiSums}
\end{equation}%
Thus, from (\ref{renyi}) and (\ref{renyiSums}),%
\begin{equation}\bigg(\frac{X_{1}^{\downarrow }}{\sum_{i=1}^{n}X_{i}},\dots ,\frac{%
X_{n}^{\downarrow }}{\sum_{i=1}^{n}X_{i}}\bigg)
\overset{d}{=}\bigg(
\frac{\sum_{j=1}^{n}\frac{X_{j}}{j}}{\sum_{i=1}^{n}X_{i}},\dots ,\frac{%
\sum_{j=n}^{n}\frac{X_{j}}{j}}{\sum_{i=1}^{n}X_{i}}\bigg),
\end{equation}%
and by (\ref{jambunathan})%
\begin{equation}
\bigg(\frac{\sum_{j=1}^{n}\frac{X_{j}}{j}}{\sum_{i=1}^{n}X_{i}},\dots ,\frac{%
\sum_{j=n}^{n}\frac{X_{j}}{j}}{\sum_{i=1}^{n}X_{i}}\bigg)
=
\bigg(\sum_{j=1}^{n}\frac{W_{j}}{j},\dots ,\sum_{j=n}^{n}\frac{W_{j}}{j}\bigg),
\end{equation}
as desired.
\end{proof}

\subsection{Product--moments of weights}

Product--moments of voting weights are interesting for a number of reasons.
Firstly, they appear in the definition of R\'{e}nyi entropy \citep{Renyi61}
of integer order $m$, where $m>1$, given by $-\ln
\sum\nolimits_{i=1}^{n}w_{i}^{m}$. Secondly, we use them in Sec. \ref%
{sec:coleman} to obtain the characteristic function of the distribution of
the total weight of a random coalition of players. Finally, the sum of
squared weights appears in the definitions of the
Herfindahl--Hirschman--Simpson index of diversity \citep
{Hirschman45,Simpson49,Herfindahl50}, $\sum\nolimits_{i=1}^{n}w_{i}^{2}$,
the Laakso--Taagepera index of the effective number of players \citep%
{LaaksoTaagepera79,TaageperaGrofman81}, $\left(
\sum\nolimits_{i=1}^{n}w_{i}^{2}\right) ^{-1}$, and the optimal quota
minimizing the Euclidean distance between weight and power vectors \citep%
{SlomczynskiZyczkowski06,SlomczynskiZyczkowski07}, $\frac{1}{2}\left(
1+(\sum_{i=1}^{n}w_{i}^{2})^{-1}\right) $.

We obtain a general theorem about the expected value of the product--moment
of voting weights:

\begin{theorem}
\label{thm:moments}If $\mathbf{W}\sim \func{Unif}\left( \Delta _{n}\right) $%
, then for every $\mathbf{m}:=\linebreak \left( m_{1},\dots ,m_{n}\right) \in \mathbb{N}%
^{n}$,%
\begin{equation}
\mathbb{E}%
\bigg(\prod_{j=1}^{n}W_{j}^{m_{j}}\bigg)
=\frac{\prod_{j=1}^{n}m_{j}!}{\left( n\right) _{\left\vert m\right\vert }},
\label{productMoments}
\end{equation}%
where $\left\vert \mathbf{m}\right\vert :=\sum_{j=1}^{n}m_{j}$ and $\left(
k\right) _{l}:=\prod_{j=0}^{l-1}(k+j)$.\medskip
\end{theorem}

\begin{proof}
Substituting $d=n-1$, $D=n$, and $%
l_{j}=x_{j}$ ($j=1,\dots ,n$) for $\left( t_{1},\dots
,t_{n}\right) \in \left( 0,1\right) ^{n}$ in \citet[Corollary 14]{BaldoniEtAl11} we obtain%
\begin{equation}
\sum_{\mathbf{m}\in \mathbb{N}^{n}}\prod_{j=1}^{n}t_{j}^{m_{j}}\frac{\left(
n\right) _{\left\vert \mathbf{m}\right\vert }}{\prod_{j=1}^{n}\left(
m_{j}\right) !}\int_{\Delta _{n}}\prod_{j=1}^{n}x_{j}^{m_{j}}\,d\mathbf{x}=%
\frac{1}{\prod_{j=1}^{n}\left( 1-t_{j}\right) }.
\end{equation}%
Expanding $\left( 1-t_{j}\right) ^{-1}$ into Taylor series, we get%
\begin{equation}
\sum_{\mathbf{m}\in \mathbb{N}^{n}}\prod_{j=1}^{n}t_{j}^{m_{j}}\frac{\left(
n\right) _{\left\vert \mathbf{m}\right\vert }}{\prod_{j=1}^{n}\left(
m_{j}\right) !}\,\mathbb{E}%
\bigg(\prod_{j=1}^{n}W_{j}^{m_{j}}\bigg)
=\sum_{\mathbf{m}\in \mathbb{N}^{n}}\prod_{j=1}^{n}t_{j}^{m_{j}}.
\end{equation}%
Then the assertion follows from the uniqueness of Taylor expansion.\medskip
\end{proof}

From this result, we obtain the following corollaries:

\begin{corollary}
If a random vector $\mathbf{W}\sim \func{Unif}\left( \Delta _{n}\right) $,
then for every $m\in \mathbb{N}_{+}$,%
\begin{equation}
\mathbb{E}\bigg(\sum_{j=1}^{n}W_{j}^{m}\bigg)
=\frac{m!}{\left( n+1\right) _{m-1}}.
\end{equation}
\end{corollary}

\begin{corollary}
\label{thm:ssq}If a random vector $\mathbf{W}\sim \func{Unif}\left( \Delta
_{n}\right) $, then:%
\begin{equation}
\mathbb{E}\bigg(\sum_{j=1}^{n}W_{j}^{2}\bigg)=\frac{2}{n+1},
\end{equation}%
and%
\begin{equation}
\func{Var}\bigg(\sum_{j=1}^{n}W_{j}^{2}\bigg)=
\frac{4\left( n-1\right) }{\left( n+1\right) ^{2}\left( n+2\right) \left(
n+3\right) }.
\end{equation}
\end{corollary}

\section{Voting Power of the $k$--th Largest Player}

\subsection{Definitions}

The notion of a \emph{power index} serves to characterize the \textit{a
priori} voting power of a player in a weighted voting game by measuring the
probability that their vote will be decisive in a hypothetical ballot, i.e.,
the winning coalition would fail to satisfy the qualified majority condition
if this player were to change their vote. In the classical approach by
\citet{Penrose46,Penrose52} and \citet{Banzhaf64}, it is
assumed that all potential coalitions of players are equiprobable.

Let $\omega :=\left\vert \mathcal{W}\right\vert $ be the total number of
winning coalitions, and for $i=1,\dots ,n$, let $\omega _{i}:=\left\vert
\left\{ Q\in \mathcal{W}:i\in Q\right\} \right\vert $ be the number of
winning coalitions that include the $i$--th player.

\begin{definition}
The absolute (non-normalized) Penrose--Banzhaf index $\psi _{i}$ of the $i$%
--th player, where $i=1,\ldots ,n$, is the probability that the $i$--th
player is decisive, i.e., 
\begin{equation}
\psi _{i}:=\frac{\omega _{i}-(\omega -\omega _{i})}{2^{n-1}}=\frac{2\omega
_{i}-\omega }{2^{n-1}}.  \label{PB1}
\end{equation}
\end{definition}

To compare these indices for games with different numbers of players, it is
convenient to define the \textit{normalized Penrose--Banzhaf index}.

\begin{definition}
The normalized Penrose--Banzhaf index $\beta _{i}$ of the $i$--th player,
where $i=1,\ldots ,n$, is 
\begin{equation}
\beta _{i}:=\frac{\psi _{i}}{\sum\nolimits_{j=1}^{n}\psi _{j}}.  \label{PB2}
\end{equation}
\end{definition}

The absolute Penrose--Banzhaf index, unlike the normalized one, has a clear
probabilistic interpretation; however, for the latter the vector of indices
always lies on $\Delta _{n}$.

\subsection{Analytical results for very small values of $n\label%
{sec:analyticalSmallN}$}

For any $G,J\in \mathcal{G}_{n}$, let $I:V\left( G\right) \rightarrow V\left(
J\right) $ be an isomorphism mapping the $k$--th largest player in $G$ to
the $k$--th largest player in $J$ (assuming linear orderings of players in
both games), and let $\sim $ be an equivalence relation on $\mathcal{G}_{n}$
such that $G\sim J$ if and only if $\mathcal{W}\left( G\right) =\mathcal{W}%
\left( J\right) $ up to isomorphism $I$. For small values of $n$, the
elements of the quotient set $\mathcal{G}_{n}/\sim $ can be easily
enumerated -- see \citet{MurogaEtAl62,Winder65,MurogaEtAl70} and more
generally \citet{KirschLangner10,BarthelemyEtAl11,Kurz12,Kurz18c}. Their
number increases rapidly with $n$: there are $2$ elements of $\mathcal{G}%
_{n}/\sim $ for $n=2$ players, $5$ for $3$ players, $14$ for $4$
players, $62$ for $5$ players, $566$ for $6$ players, and $11971$ for $7$
players.

For a fixed $q\in (\frac{1}{2},1]$ and for each $\chi \in \mathcal{G}%
_{n}/\sim $ there exists a set $L_{\chi }^{q}\subset \Delta _{n}$\ such that
for any point within $L_{\chi }^{q}$ the ordered power index vector $(\beta
_{1}^{\downarrow },...,\beta _{n}^{\downarrow })$ equals $(\beta _{1}^{\chi
},...,\beta _{n}^{\chi })$. Note that the volume of $L_{\chi }^{q}$ depends
on the quota $q$. The expected voting power of the $k$--th largest player
equals:%
\begin{equation}
\mathbb{E}\left( \beta _{k}^{\downarrow }\right) =\sum_{\chi \in \mathcal{G}%
_{n}/\sim }\beta _{k}^{\chi }\,\lambda \left( L_{\chi }^{q}\right) ,
\label{eq:expectBetaI}
\end{equation}%
where by $\lambda $ we denote the Lebesgue measure on $\Delta _{n}$.

The case of $n=2$ is straightforward, as there are only two classes of games
-- the unanimity and the dictatorship of the largest player. Ordered power
index vectors $(\beta _{1}^{\downarrow },\beta _{2}^{\downarrow })$ for
those classes are equal to $(\frac{1}{2},\frac{1}{2})$ and $(1,0)$,
respectively. Thus, we obtain:
\begin{subequations}
\label{eq:expectBetaOrd}
\begin{equation}
\mathbb{E}\left( \beta _{1}^{\downarrow }\right) =2 \left(\frac{1}{2} \lambda \Big(\big(\frac{1}{2}, q\big)\Big) +
\lambda\Big(\big(q, 1\big)\Big)\right)=\frac{3}{2}-q,
\end{equation}\begin{equation}
\mathbb{E}\left( \beta _{2}^{\downarrow }\right) =\lambda\Big(\big(\frac{1}{2}, q\big)\Big)=q-\frac{1}{2}.
\end{equation}\end{subequations}%

Now let us consider the simplest non-trivial case -- that of $n=3$. There
are five elements of $\mathcal{G}_{n}/\sim $ to consider:\medskip

\begin{adjustbox}{max width=\columnwidth}
\renewcommand*{\arraystretch}{1.2}
\begin{tabular}{|c|c|c|c|}
\hline
$\beta^{\chi }$ & condition ($\chi $) & probability ($\lambda \left( L_{\chi }^{q}\right) $)
\\ \hline
$(\frac{1}{3},\frac{1}{3},\frac{1}{3})$ & $q>w_{1}^{\downarrow }+w_{2}^{\downarrow }$ & $%
1-F_{3}(1-q)$ \\ \hline
$(\frac{1}{2},\frac{1}{2},0)$ & $w_{1}^{\downarrow }+w_{2}^{\downarrow
}>q>w_{1}^{\downarrow }+w_{3}^{\downarrow }$ & $F_{3}(1-q)-F_{2}(1-q)$ \\ 
\hline
$(\frac{3}{5},\frac{1}{5},\frac{1}{5})$ & $%
\begin{array}{c}
(w_{1}^{\downarrow }+w_{3}^{\downarrow }>q>w_{2}^{\downarrow
}+w_{3}^{\downarrow }) \\ 
\wedge \,(w_{1}^{\downarrow }+w_{3}^{\downarrow }>q>w_{1}^{\downarrow })%
\end{array}%
$ & $\!\!\!%
\begin{array}{c}
(1-F_{1}(1-q))F_{1}(q) \\ 
-F_{1}(1-q)(1-F_{1}(q)) \\ 
-1+F_{2}(1-q)%
\end{array}%
\!\!\!$ \\ \hline
$(\frac{1}{3},\frac{1}{3},\frac{1}{3})$ & $\begin{array}{c}(w_{2}^{\downarrow }+w_{3}^{\downarrow
}>q) \\ \wedge \,(w_{1}^{\downarrow }<1/2)\end{array}$ & $F_{1}(1-q)$ \\ \hline
$(1,0,0)$ & $(w_{1}^{\downarrow }>q)\wedge (w_{1}^{\downarrow
}>1/2)$ & $1-F_{1}(q)$ \\ \hline
\end{tabular}
\end{adjustbox}

\begin{flushleft}
where by $F_{k}$ we denote the cumulative distribution of the $k$--th
largest player's weight, $k=1,2,3$.
\end{flushleft}

From the above and (\ref{densKthPlayer}), we obtain, see Fig. \ref{fig:powerDistProbs}:\medskip

\begin{adjustbox}{max width=\columnwidth}
\renewcommand*{\arraystretch}{1.2}
\begin{tabular}{c|c|c|}
\cline{2-3}
 & \multicolumn{2}{|c|}{$\qquad \lambda \left( L_{\chi }^{q}\right) $%
} \\ \hline
\multicolumn{1}{|c|}{$\beta^{\chi }$} & $\quad \qquad q\leq 2/3\qquad \quad $ & $\quad \qquad q\geq
2/3\qquad \quad $ \\ \hline
\multicolumn{1}{|c|}{$(\frac{1}{3},\frac{1}{3},\frac{1}{3})$} & $0$ & $9q^{2}-12q+4$ \\ \hline
\multicolumn{1}{|c|}{$(\frac{1}{2},\frac{1}{2},0)$} & $12q^{2}-12q+3$ & $-15q^{2}+24q-9$ \\ 
\hline
\multicolumn{1}{|c|}{$(\frac{3}{5},\frac{1}{5},\frac{1}{5})$} & $-24q^{2}+30q-9$ & $3q^{2}-6q+3$ \\ 
\hline
\multicolumn{1}{|c|}{$(\frac{1}{3},\frac{1}{3},\frac{1}{3})$} & $9q^{2}-12q+4$ & $0$ \\ 
\hline
\multicolumn{1}{|c|}{$(1,0,0)$} & $3q^{2}-6q+3$ & $3q^{2}-6q+3$ \\ \hline
\end{tabular}
\end{adjustbox}
\medskip

At this point, from (\ref{eq:expectBetaI}) we get:
\begin{subequations} \label{eq:expectBeta123}
	\begin{alignat}{4}
		\mathbb{E}\left(\beta_1^\downarrow\right)&=\begin{cases}
			\frac{12}{5}(q-q^2)+\frac{1}{30}, & q \leq 2/3, \\
			\frac{1}{10}(16-16q+3q^2)+\frac{1}{30}, & q \geq 2/3,
		\end{cases} \\
		\mathbb{E}\left(\beta_2^\downarrow\right)&=\begin{cases}
			\frac{21}{5}q^2-4q+1+\frac{1}{30}, & q \leq 2/3, \\
			\frac{1}{10}(68q-39q^2+26), & q \geq 2/3,
		\end{cases} \\
		\mathbb{E}\left(\beta_3^\downarrow\right)&=\begin{cases}
			-\frac{9}{5}q^2+2q-\frac{1}{2}+\frac{1}{30}, & q \leq 2/3, \\
			\frac{2}{5}(9q^2-13q)+2+\frac{1}{30}, & q \geq 2/3.
		\end{cases}
	\end{alignat}
\end{subequations}%

\begin{figure}[h]
	\centering
	\includegraphics[width=\linewidth]{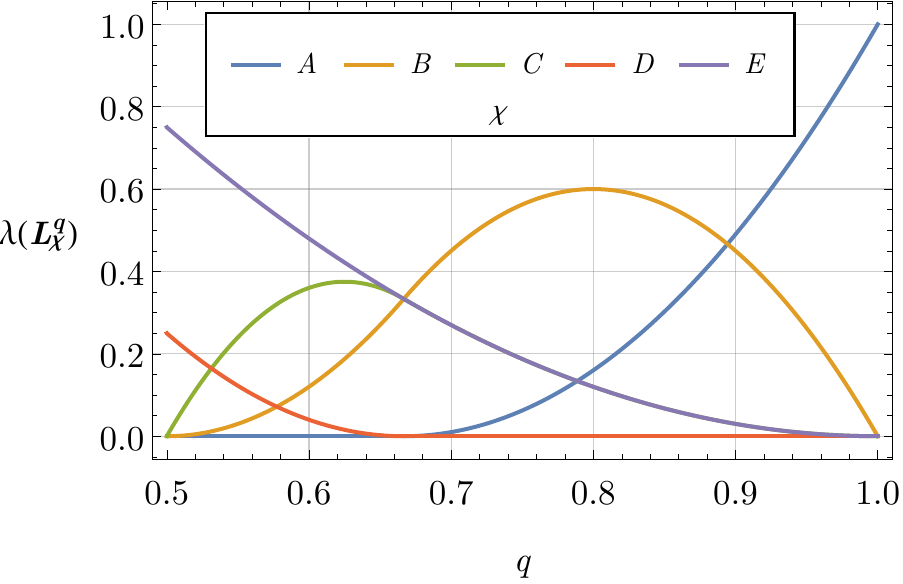}
	\caption{Probabilities of the five classes of games $\chi \in \mathcal{G}_{3}/\sim $
		as a function of the quota $q$. Classes $A$, $B$, $C$, $D$, and $E$ correspond,
		respectively, to weight vectors $\beta^{A} := (\frac{1}{3},\frac{1}{3},\frac{1}{3})$,
		$\beta^{B} := (\frac{1}{2},\frac{1}{2},0)$,
		$\beta^{C} := (\frac{3}{5},\frac{1}{5},\frac{1}{5})$,
		$\beta^{D} := (\frac{1}{3},\frac{1}{3},\frac{1}{3})$,
		and $\beta^{E} := (1,0,0)$.
	}
	\label{fig:powerDistProbs}
\end{figure}%

\subsection{Numerical results for small values of $n\label{sec:numericalBan}$%
}

As mentioned in Sec. 1, if a player's voting weight is fixed, the dependence
of the voting power on the quota $q\in (\frac{1}{2},1]$ seems to be highly
erratic. This is illustrated by Figure \ref{fig:sixtq}.

\begin{figure}[t]
	\centering
	$\vcenter{\hbox{\includegraphics[width=\linewidth]{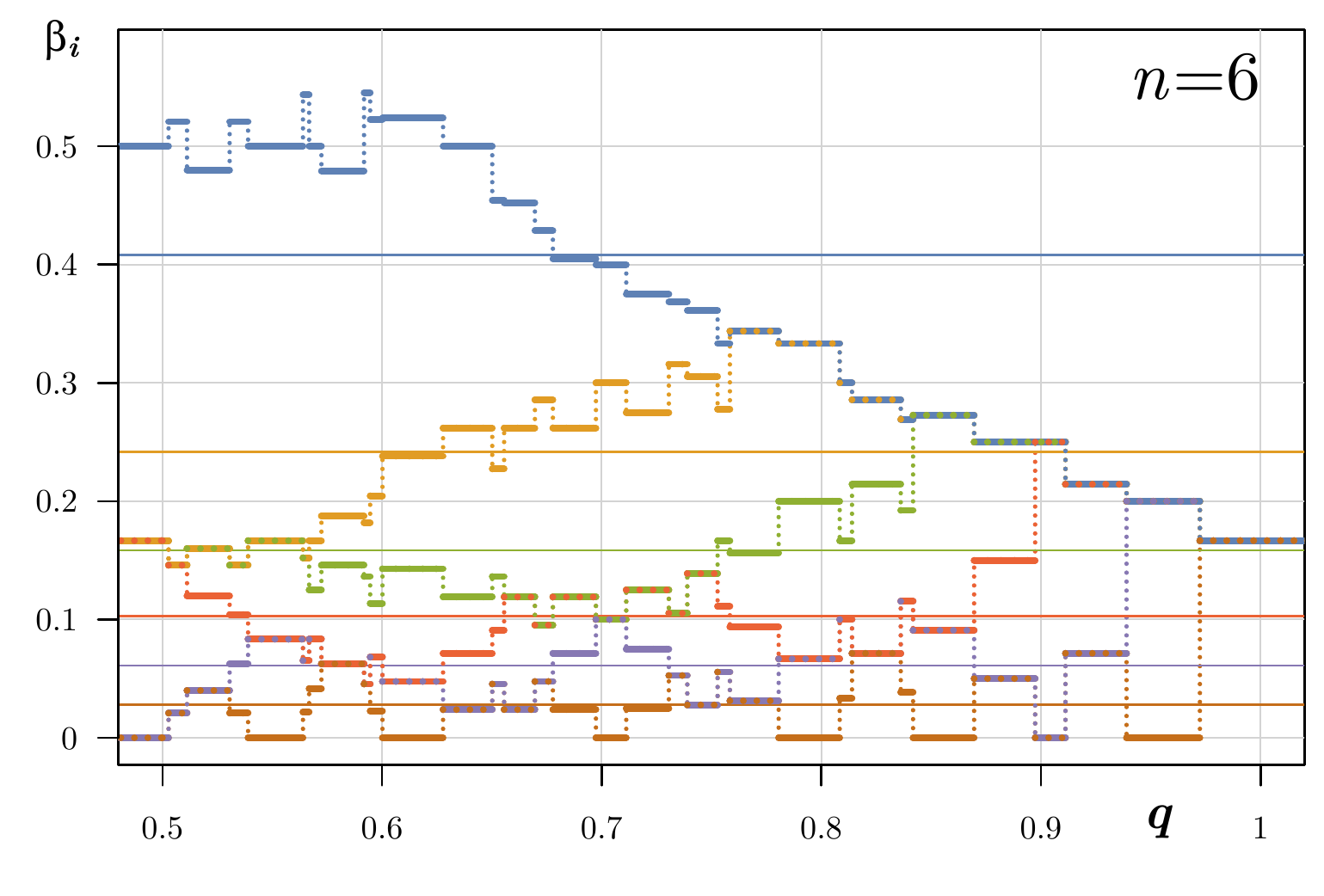}}}$
	$\vcenter{\hbox{\includegraphics[width=0.6\linewidth]{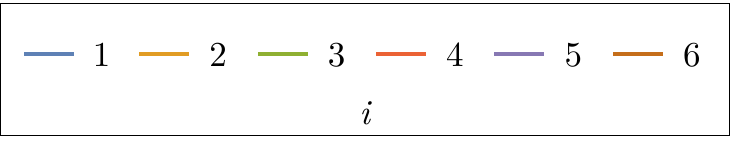}}}$
	\caption{Normalized Penrose--Banzhaf power indices $\beta _{1},\dots ,\beta _{6}$
		in a six-player weighted voting game with weights fixed at the barycenter of the asymmetric simplex, 
		$\QTR{bf}{b}=(147,87,57,37,22,10)/360$, as functions of the quota $q$.
		Horizontal lines represent the voting weights of each player. An earlier version of this figure appeared in \citet[p.~282]{RzazewskiEtAl14},
		cf. Fig. \ref{fig:sixrand}.}
	\label{fig:sixtq}
\end{figure}%

On Fig. \ref{fig:sixrand} we plot numerical estimates of $\mathbb{E(}\beta
_{k}^{\downarrow })$ and $\mathbb{E(}\psi _{k}^{\downarrow })$\ as functions
of $q$, obtained by Monte Carlo samplings of $2^{16}$ random vectors of
length $n=3,6,9$. Their examination reveals certain general regularities.

\begin{figure*}
	\centering
	\includegraphics[width=0.48\linewidth]{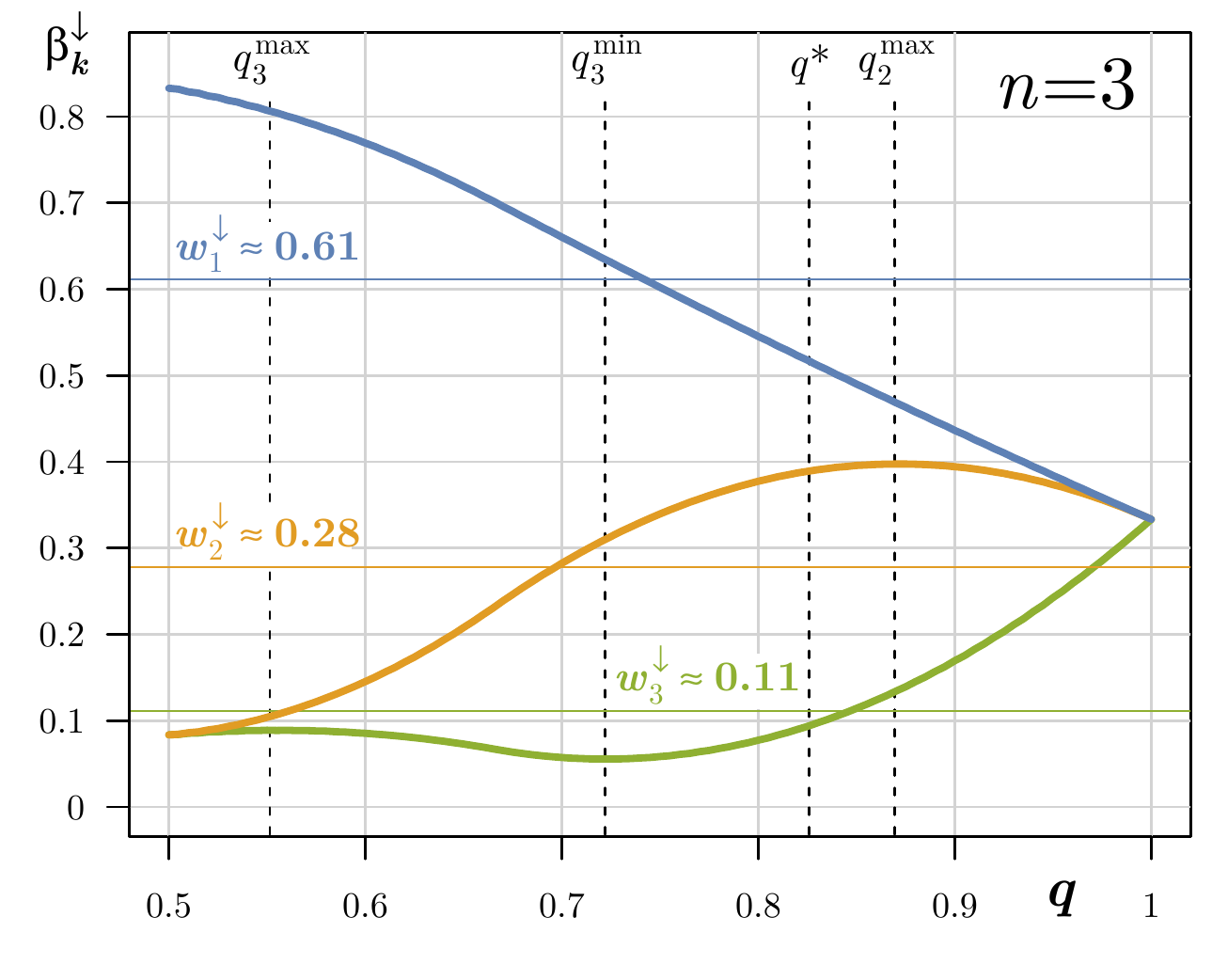}
	\includegraphics[width=0.48\linewidth]{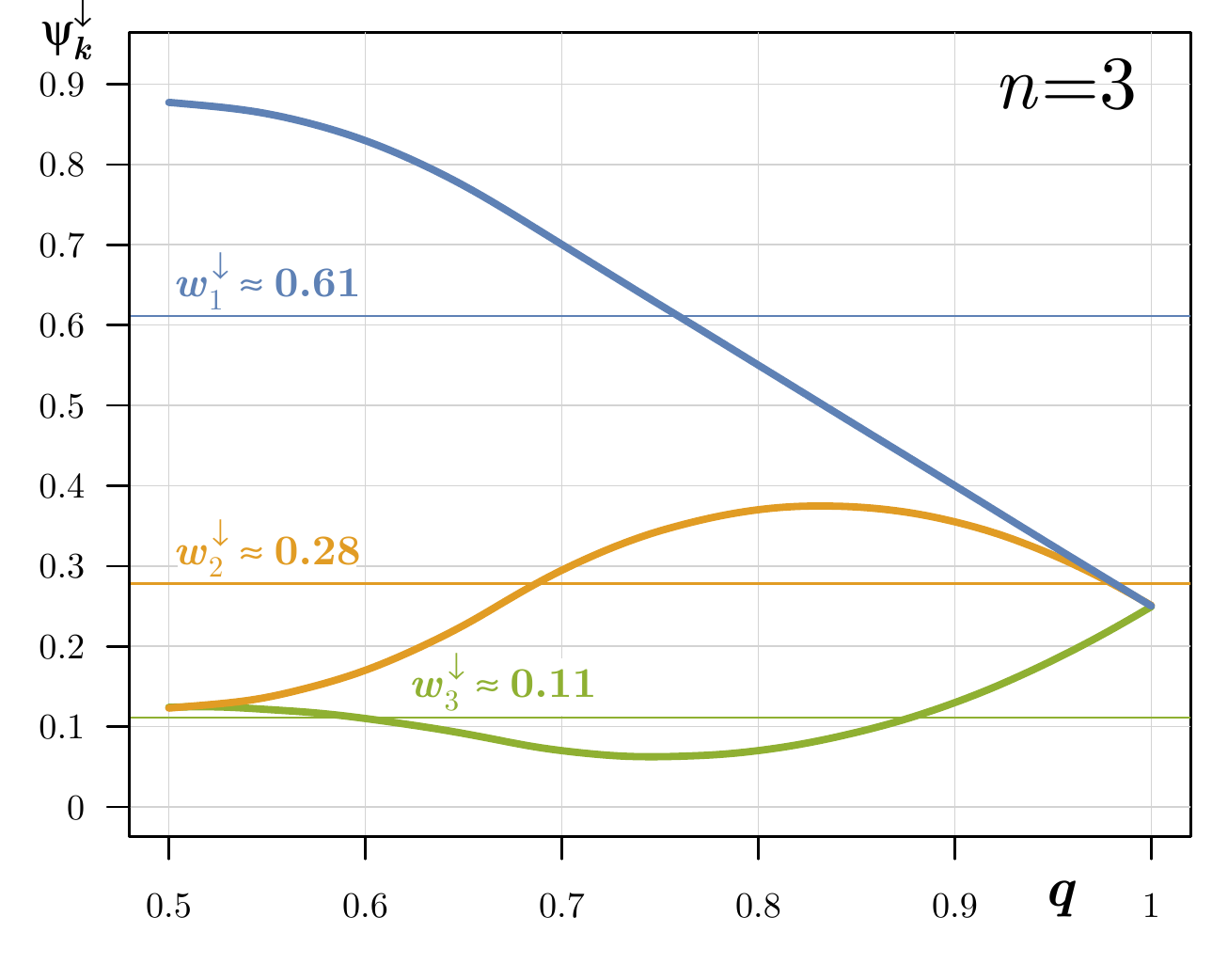}
	\includegraphics[width=0.48\linewidth]{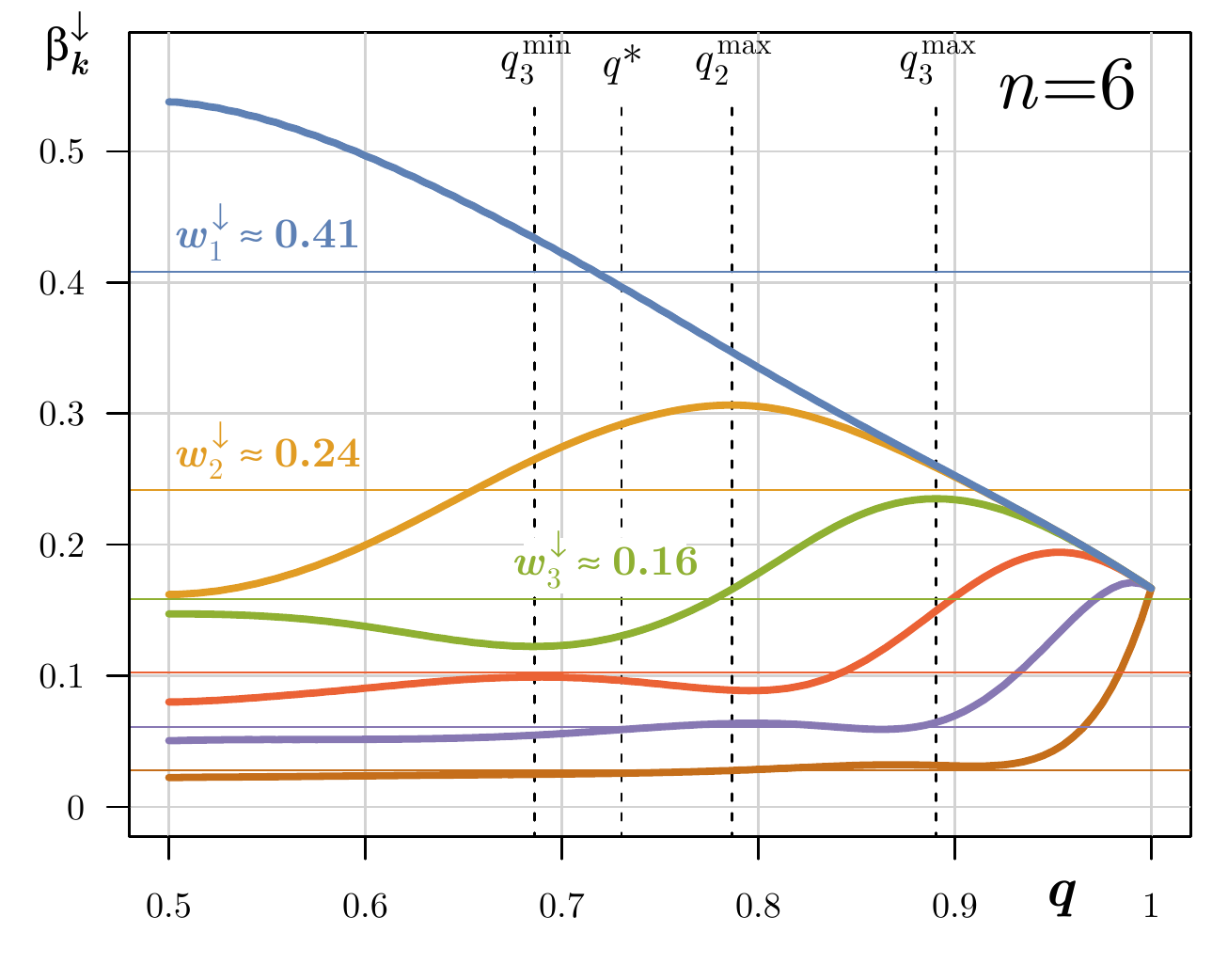}
	\includegraphics[width=0.48\linewidth]{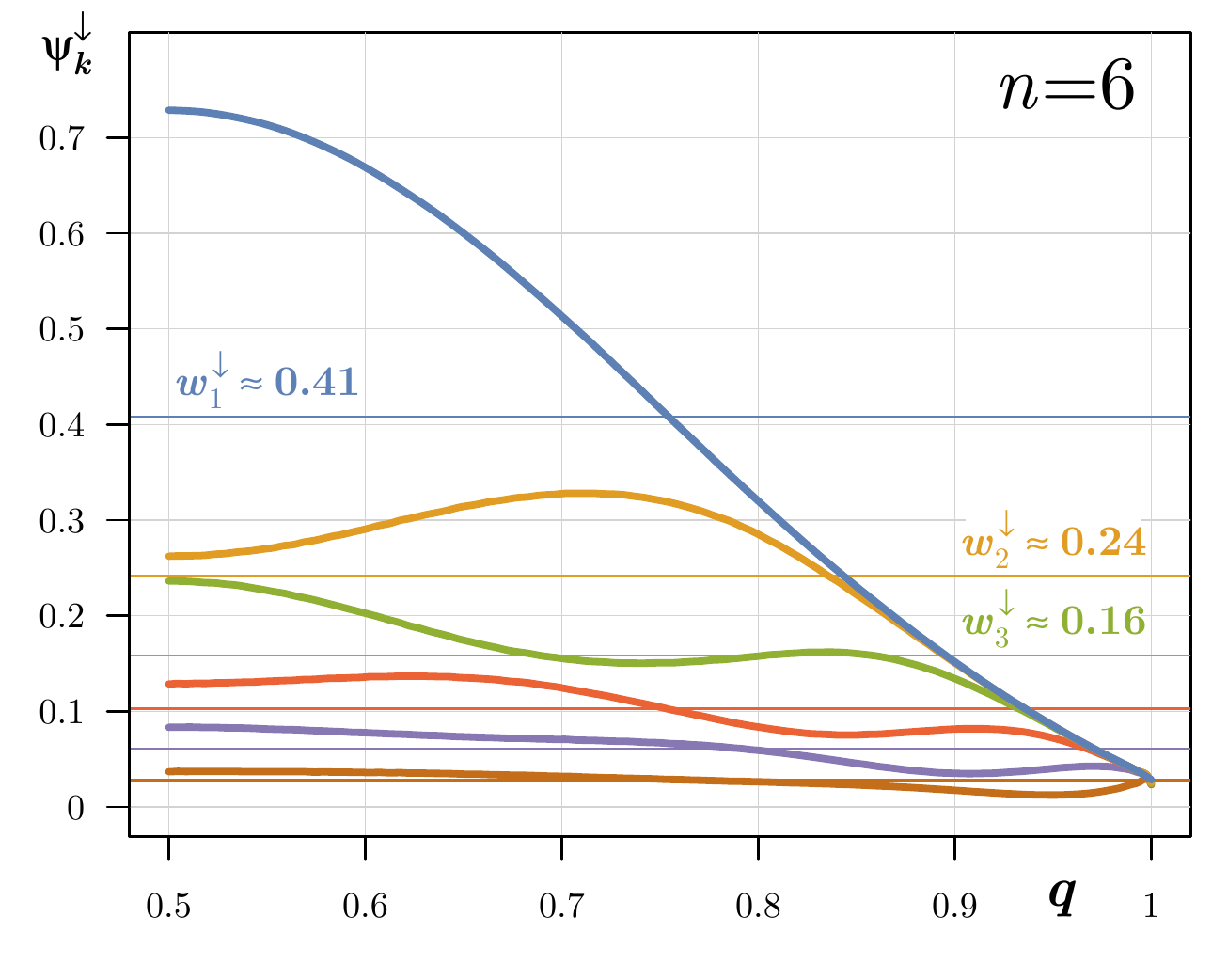}
	\includegraphics[width=0.48\linewidth]{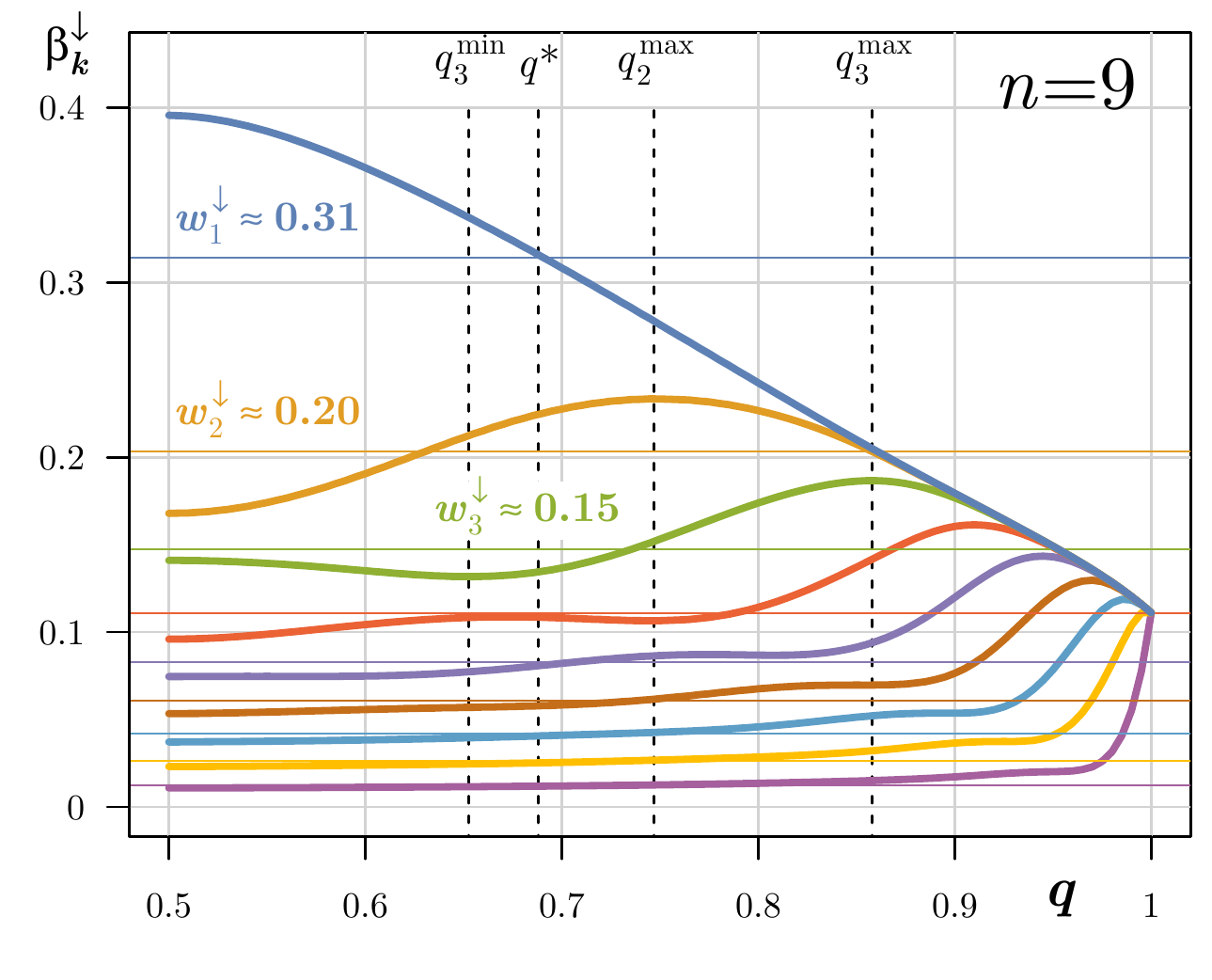}
	\includegraphics[width=0.48\linewidth]{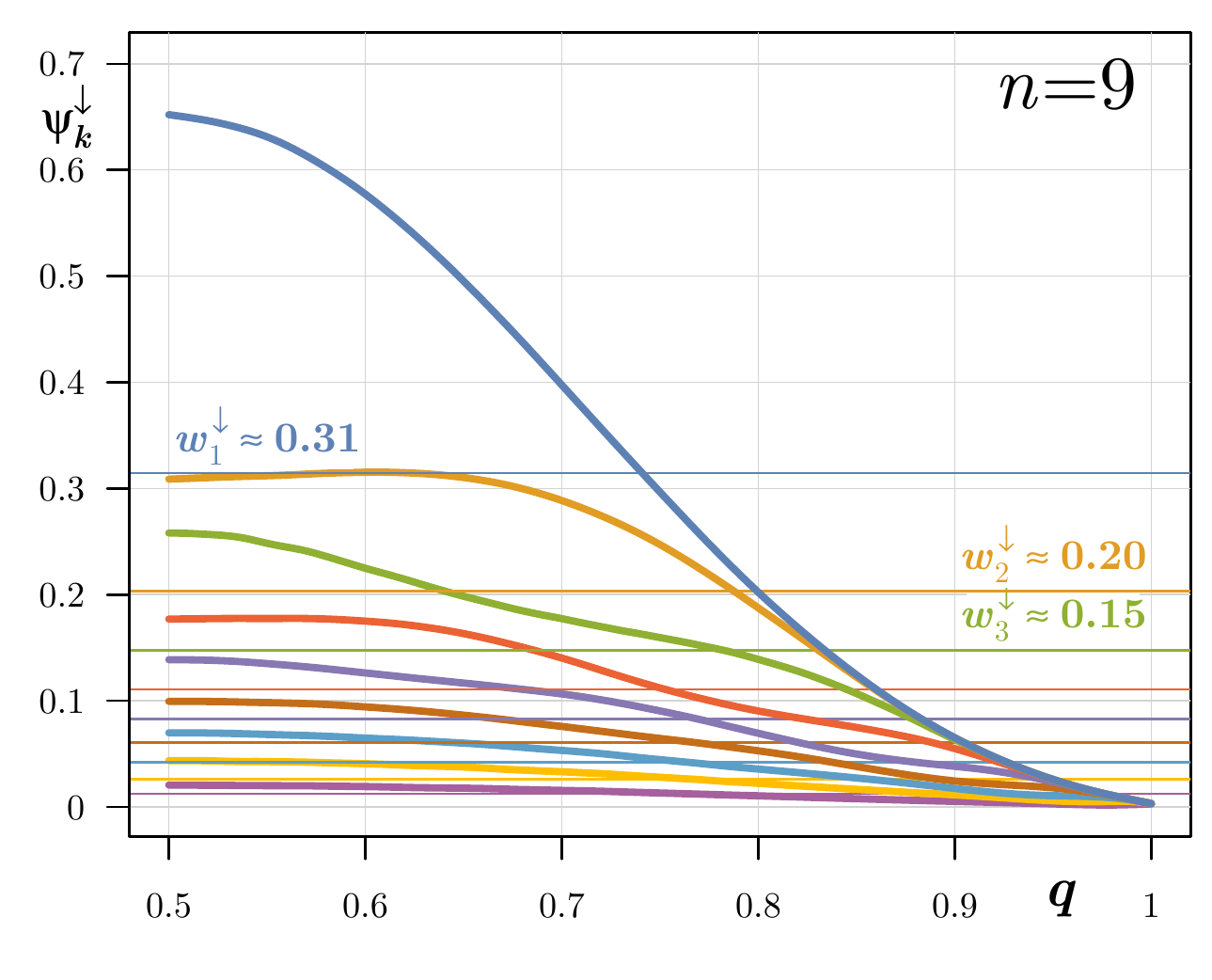}
	\caption{Absolute and normalized Penrose--Banzhaf power indices of $n$ players averaged
		over the probability simplex $\Delta_n$ with respect to the uniform measure as functions of the quota $q$.
		Horizontal lines represent the average voting weight of each player. The vertical line $q = q^*$
		represents the approximation of the quota minimizing the distance $\left\Vert \mathbf{w}-\boldsymbol{\beta }\right\Vert _{2}$, see \citet{ZyczkowskiSlomczynski13}. An earlier version of one of the figures appeared in \citet[p.~287]{RzazewskiEtAl14}.}
	\label{fig:sixrand}
\end{figure*}

For $q\rightarrow 1/2$ the average voting power of the largest player, $%
\mathbb{E(}\beta _{1}^{\downarrow })$, is considerably greater than their
average weight, $\mathbb{E(}w_{1}^{\downarrow })$, at the expense of all the
other players, and then decreases monotonically with the quota $q$. The
second player initially loses the most, but their average voting power, $%
\mathbb{E(}\beta _{2}^{\downarrow })$, increases up to its single maximum, $%
q_{2}^{\max }$, while the average voting power of the third player, $\mathbb{%
E(}\beta _{3}^{\downarrow })$, has two extrema, $q_{3}^{\min }$ and $%
q_{3}^{\max }$. The average voting power of small players initially
fluctuate mildly with $q$ around their average weights, with the amplitudes
of these fluctuations diminishing as $k$ increases, and for $q\rightarrow 1$
the voting powers of all players converge to $1/n$.

\medskip Careful examination of the numerical results suggests a following
conjecture:

\begin{conjecture}
\label{con:critPoints} For the uniform distribution on the probability
simplex $\Delta _{n}$ and for every $k=1,...,n$, the average normalized
Penrose--Banzhaf power index of the $k$--th largest player, $\mathbb{E}%
(\beta _{k}^{\downarrow })$, has exactly $k-1$ local extrema over $(1/2,1)$
as a function of $q$.
\end{conjecture}

\begin{remark}
Note that for $n=3$, Conjecture \ref{con:critPoints} follows immediately
from the analytic form of $\mathbb{E}(\beta _{k}^{\downarrow })$ given by
(\ref{eq:expectBeta123}). The voting power of the second largest player, 
$\mathbb{E(}\beta _{2}^{\downarrow })$, admits a maximal value at $%
q_{2}^{\max }=34/39$ $(\approx 87.18\%)$, while $\mathbb{E(}\beta
_{3}^{\downarrow })$ exhibits a minimum at $q_{3}^{\min }=5/9$ $(\approx
55.56\%)$ and a maximum at $q_{3}^{\max }=13/18$ $(\approx 72.22\%)$.
\end{remark}


\section{The power of a collectivity to act\label{sec:coleman}}

\emph{The power of a collectivity to act}, i.e., the ease of reaching a
decision, is usually measured with the \emph{Coleman efficiency index }\citep%
{Coleman71}, defined as the probability that a random coalition $Q\in 
\mathcal{P}(V)$ is a winning one:%
\begin{equation}
C:=\frac{\omega }{2^{n}},
\end{equation}%
where $\omega :=\left\vert \mathcal{W}\right\vert $.

\begin{remark}
Note that $C$ is a decreasing function of the quota $q\in (\frac{1}{2},1]$. Since it
is impossible for any coalition $Q\in \mathcal{P}(V)$ that both $Q$ and $%
V\setminus Q$ be winning, $C\leq \frac{1}{2}$. On the other hand, $C\geq
C\left( 1\right) =2^{-\left\vert\left\{ j\,=\,1,\dots ,n~:~w_{j}>0\right\}\right\vert }.$%
\smallskip 
\end{remark}

Let $\mu _{n}$ be the \emph{Bernoulli measure} on $\left\{ 0,1\right\} ^{n}$%
, and let \linebreak $Z:\Delta _{n}\times \left\{ 0,1\right\} ^{n}\rightarrow \mathbb{R}
$ be given by the formula $Z\left( \mathbf{w},\mathbf{\xi }\right)
:=\sum_{i=1}^{n}w_{i}\xi _{i}-\frac{1}{2}$, where $\mathbf{w}\in \Delta _{n}$
and $\mathbf{\xi }\in \left\{ 0,1\right\} ^{n}$. Note that%
\begin{equation}
\mathbb{E}_{\lambda \times \mu _{n}}\left( C\right) =1-F_{Z}\left( q-\frac{1%
}{2}\right) ,  \label{colemanDistZ}
\end{equation}%
where $\lambda $ is the Lebesgue measure on $\Delta _{n}$, and $F_{Z}$ is the distribution function of $Z$ with respect to the
probability measure $\lambda \times \mu _{n}$ on $\Delta _{n}\times \left\{
0,1\right\} ^{n}$.
This distribution function can be calculated by the following proposition:

\begin{theorem}
\label{colemanCF}The characteristic function of $Z$ is given by%
\begin{equation}
\varphi _{Z}\left( t\right) ={}_{1}\digamma _{2}\left( 
\begin{array}{c}
n \\ 
\frac{1}{2}+\frac{n}{2},\frac{n}{2}%
\end{array}%
;-\left( \frac{t}{4}\right) ^{2}\right) ,
\end{equation}%
for $t \in \mathbb{R}$, where $_{1}\digamma _{2}$ is a generalized hypergeometric function.
\end{theorem}

\begin{proof}
For a fixed $\mathbf{w}\in \Delta _{n}$ and $k=1,\dots ,n$, let $%
X_{k}:=w_{k}(\xi _{k}-\frac{1}{2})$ and $X:=\sum_{k=1}^{n}X_{k}$. Then for $%
t\in \mathbb{R}$,%
\begin{equation}
\varphi _{X_{k}}\left( t\right) =\frac{1}{2}\left( e^{\frac{1}{2}%
itw_{k}}+e^{-\frac{1}{2}itw_{k}}\right) =\cos \left( \frac{tw_{k}}{2}\right)
,
\end{equation}

and%
\begin{align}
\varphi _{X}\left( t\right) &=\prod_{k=1}^{n}\cos \left( \frac{tw_{k}}{2}%
\right) =\sum_{j=0}^{\infty }\frac{t^{j}}{j!}\left. \frac{d^{j}}{dt^{j}}%
\prod_{k=1}^{n}\cos \left( \frac{tw_{k}}{2}\right) \right\vert _{t=0}  \notag
\\
&=\sum_{j=0}^{\infty }\frac{t^{2j}}{\left( 2j\right) !}\sum_{j_{1}+\ldots
+j_{n}=j}\left( -1\right) ^{j}\,2^{-2j}\left( 2j\right) !\prod_{k=1}^{n}%
\frac{w_{k}^{2j_{k}}}{\left( 2j_{k}\right) !}  \notag \\
&=\sum_{j=0}^{\infty }\left( -1\right) ^{j}\left( \frac{t}{2}\right)
^{2j}\sum_{j_{1}+\ldots +j_{n}=j}\prod_{k=1}^{n}\frac{w_{k}^{2j_{k}}}{\left(
2j_{k}\right) !}.
\end{align}%
It can be shown that the resulting series is absolutely convergent. Hence,
and by Theorem \ref{thm:moments},%
\begin{align}
\varphi _{Z}\left( t\right) &=\int_{\Delta _{n}}\varphi _{X}\left( t\right)
\,d\lambda \notag \\
&=\sum_{j=0}^{\infty }\left( -1\right) ^{j}\left( \frac{t}{2}%
\right) ^{2j}\sum_{j_{1}+\ldots +j_{n}=j} \frac{\mathbb{E}\left( \prod_{k=1}^{n}W_{k}^{2j_{k}}\right)}{\prod_{k=1}^{n}\left(
2j_{k}\right) !}  \notag
\notag \\
&=\sum_{j=0}^{\infty }\left( -1\right) ^{j}\left( \frac{t}{2}\right) ^{2j}%
\frac{1}{\left( n\right) _{2j}}\binom{j+n-1}{n-1} \notag \\
&={}_{1}\digamma _{2}\left( 
\begin{array}{c}
n \\ 
\frac{1}{2}+\frac{n}{2},\frac{n}{2}%
\end{array}%
;-\left( \frac{t}{4}\right) ^{2}\right) ,
\end{align}%
as desired.\bigskip
\end{proof}

Thus by numerical inversion of the characteristic function $\varphi _{Z}$,
we can easily estimate the expected Coleman efficiency index $\mathbb{E}%
_{\lambda \times \mu _{n}}(C)$ for any quota $q\in (\frac{1}{2},1]$. The
results for a number of arbitrarily chosen values of $n$ are plotted on Fig. %
\ref{fig:coleman}.

\begin{figure}[h]
	\centering
	\includegraphics[width=\linewidth]{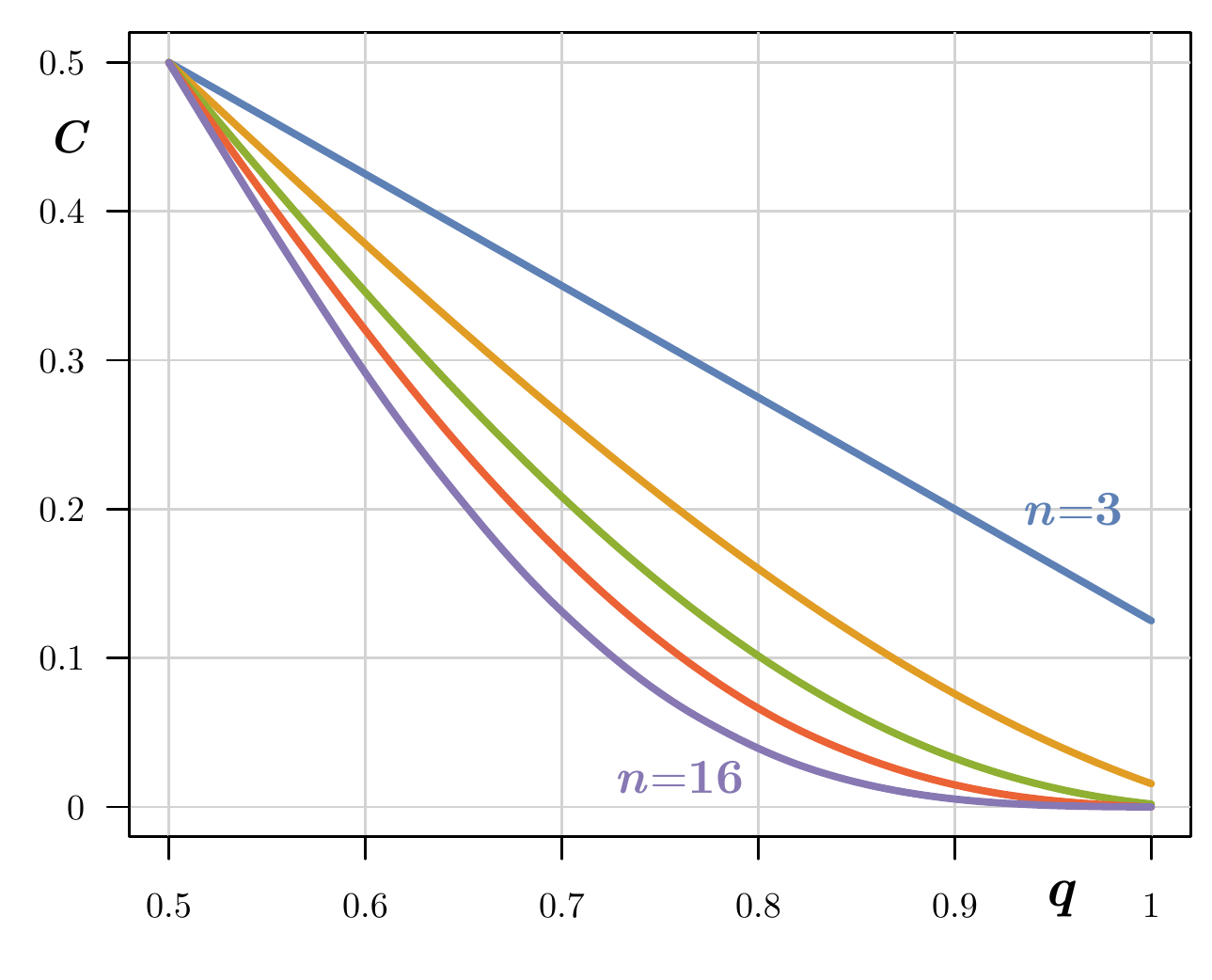}
	\caption{Coleman efficiency indices $C$ of weighted voting games with $n = 3, 6, 9, 12, 16$, averaged
		in each case over $\Delta_n$ with respect to the uniform measure, as functions of the quota $q$.}
	\label{fig:coleman}
\end{figure}%

The following results provide analytical formulae for, respectively, the
upper bound and the asymptotic approximation of the Coleman efficiency index.

\begin{remark}
Let $\mathbf{W}\sim \func{Unif}\left( \Delta _{n}\right) $. By the central
limit theorem and (\ref{colemanDistZ}), the expected Coleman
efficiency index, $\mathbb{E}_{\lambda \times \mu _{n}}\left( C\right) $,
can be approximated for fixed $n$ and $q$ by%
\begin{equation}
C_{1}:=1-\Phi \left( \sqrt{2\left( n+1\right) }\left( q-\frac{1}{2}\right)
\right) ,  \label{colemanNorm}
\end{equation}%
where $\Phi $ is the standard normal cumulative distribution function.\ The
upper bound for the approximation error can be obtained from the
Berry--Esseen theorem \citep{Berry41,Esseen42}. However, numerical
simulations suggest that it exceeds the actual approximation error by
several orders of magnitude.
\end{remark}

The above approximation is
particularly useful when one is interested in finding such value of $q$ as
to obtain a specific expected Coleman efficiency index, see Fig.  \ref{fig:colemanErr}.

\begin{figure}[h]
	\centering
	\begin{overpic}[width=\linewidth]{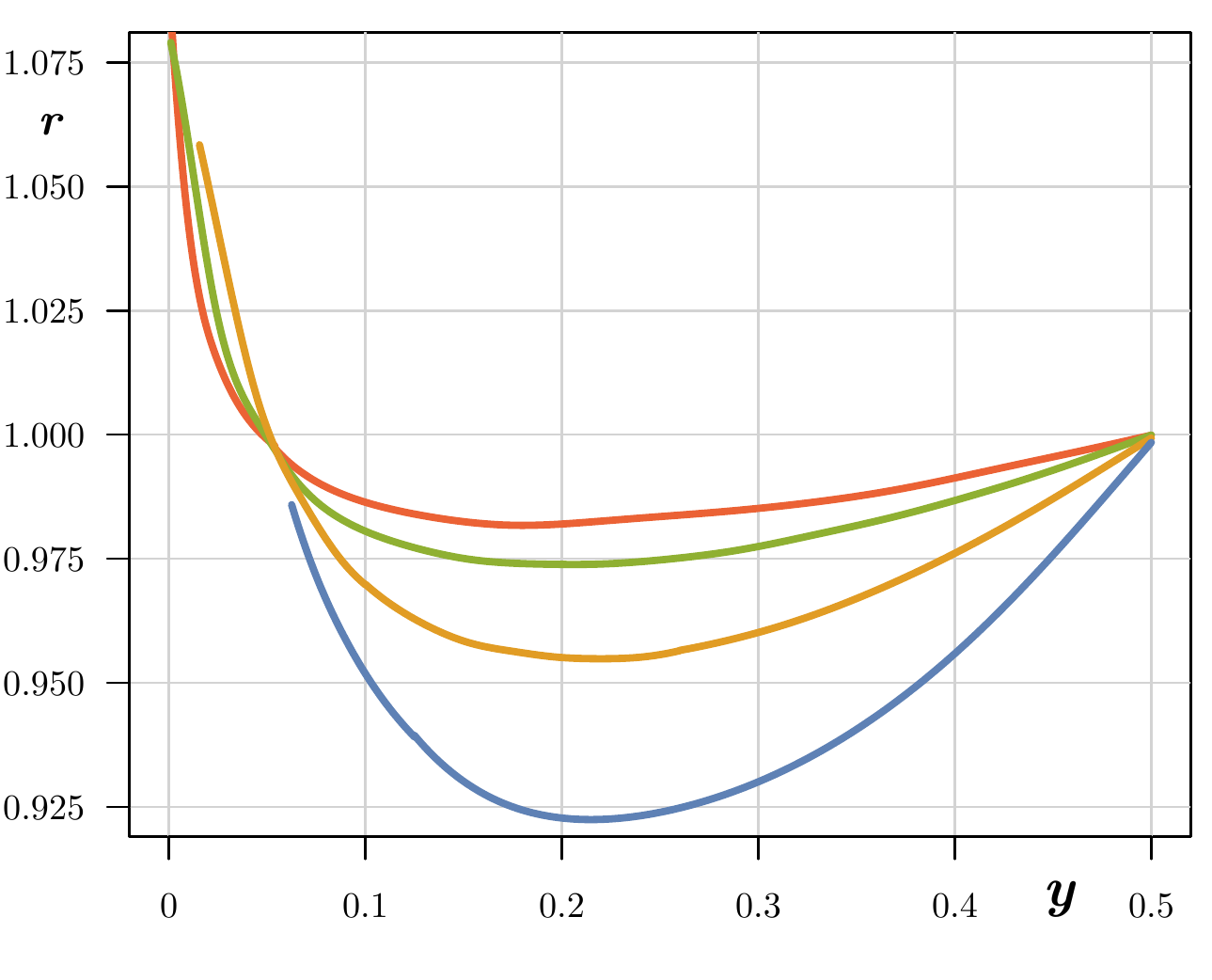}
		\put(45,62){\includegraphics[width=4.5cm]{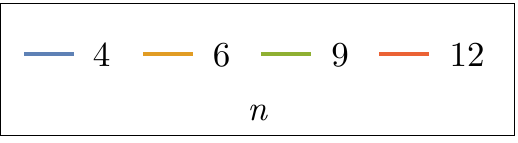}}
	\end{overpic}
	\vspace{-0.6cm}
	\caption{Error ratio $r$ of the approximation (\ref{colemanNorm}) of the expected Coleman efficiency index
		in random weighted voting games with $n = 4, 6, 9, 12$, where
		$r(y) := C_{1}^{-1}(y) / (\mathbb{E}(C))^{-1}(y) = C^{-1}_1 (\mathbb{E}(C)(q)) / q$
		for $y = \mathbb{E}(C)(q) \in [2^{-2n}, 2^{-1}]$.}
	\label{fig:colemanErr}
\end{figure}%

For any fixed weight vector $\mathbf{w}$, we have the following upper bound for the
Coleman efficiency index $C$:

\begin{proposition}
\label{thm:colemanHoef}In a weighted voting game with $q\geq 1/2$ the
Coleman efficiency index $C$ is bounded from the above in the following
manner:%
\begin{equation}
C\leq \exp \left( -\frac{2\left( q-1/2\right) ^{2}}{\sum_{i=1}^{n}w_{i}^{2}}%
\right) .
\end{equation}
\end{proposition}
 
\begin{proof}
A proof follows from the Hoeffding's inequality \citep{Hoeffding63}. If $%
Y_{1},\dots ,Y_{n}$ are independent random variables such that $Y_{i}$ is
almost surely bounded by $\left[ \tau _{i}^{-},\tau _{i}^{+}\right] $ for
every $i=1,\dots ,n$, then for any $h\geq 0$:%
\begin{equation}
\Pr \left( \sum_{i=1}^{n}\left(Y_{i}-\mathbb{E}\left( Y_{i}\right)\right)
\geq h\right) \leq \exp \left( \frac{-2h^{2}}{\sum_{i=1}^{n}\left( \tau
_{i}^{+}-\tau _{i}^{-}\right) ^{2}}\right).  \label{hoeffding}
\end{equation}%
Putting $\Pr =\mu _{n}$, $Y_{i}=w_{i}\xi _{i}$, $h=q-1/2$, $\tau
_{i}^{+}=w_{i}$ and $\tau _{i}^{-}=0$, we obtain Proposition \ref%
{thm:colemanHoef}.
\end{proof}

\section{Splines}

Any quantity that is a function of the weighted voting game (e.g., the Coleman efficiency index, the Penrose--Banzhaf and Shapley--Shubik power indices, etc.), averaged over the probability simplex $\Delta _{n}$, and considered as a function of the quota, has the following property:

\begin{theorem}
\label{thm:spline}Let $U:\Delta _{n}\times (\frac{1}{2},1]\rightarrow 
\mathcal{G}_{n}$ be a function mapping a weight vector and a quota to the related
weighted voting game. If $\mathbf{W}\sim \func{Unif}\left( \Delta
_{n}\right) $, then for any $p:\mathcal{G}_{n}\rightarrow \mathbb{R}$,%
\begin{equation}
(1/2,1]\ni q\rightarrow \mathbb{E}\left( p\left( U\left( \mathbf{W}%
,q\right) \right) \right) \in \mathbb{R}
\end{equation}%
is a spline of degree at most $n-1$.
\end{theorem}

\begin{proof}
Note that $\left( \Delta _{n}\times (1/2,1]\right) /\ker U$ is a
partition of the polytope $\Delta _{n}\times (1/2,1]$ into blocks
\begin{equation}
P_{G}:=\left\{ (\mathbf{w,}q\mathbf{)}\in \Delta _{n}\times (1/2,1]:U\left( \mathbf{w},q\right) =G\right\}
\end{equation}for $G\in \mathcal{G}_{n}$.
Each $P_{G}$ is a convex polytope \citep[ch. 2]{GrunbaumEtAl03}, since it can
be described by a system of $2^{n}$ linear inequalities with one inequality
for each coalition $Q\in \mathcal{P}(V)$, corresponding to the condition
that $Q$ be winning or losing, i.e., that $\left\langle \boldsymbol{1}_{Q},%
\mathbf{w}\right\rangle \geq q$ if $Q\in G$ and $\left\langle \boldsymbol{1}%
_{Q},\mathbf{w}\right\rangle \leq q$ if $Q\notin G$ \citep{MasonParsley16}.
For every $G\in \mathcal{G}_{n}$, and $q\in (\frac{1}{2},1]$, the
intersection of $P_{G}$ and an affine hyperplane $\Theta _{q}:=\{\mathbf{x}%
\in \mathbb{R}^{n}:\left\langle \mathbf{1},\mathbf{x}\right\rangle
=1\}\times \{q\}$ parallel to $\Delta _{n}$, is called the \emph{weight
polytope} $P_{G}^{q}$\emph{\ }\citep{Kurz18a}.

For any $p:\mathcal{G}_{n}\rightarrow \mathbb{R}$, let $q\in (1/2,1]$
be fixed. Clearly, $p\left( U\left( \mathbf{W},q\right) \right) $ is
constant over $P_{G}^{q}$ for each $G\in \mathcal{G}_{n}$. Thus, $\mathbb{E}
\left( p\left( U\left( \mathbf{W},q\right) \right) \right) $ is an affine
combination of the volumes of weight polytopes:%
\begin{equation}
\mathbb{E}\left( p\left( U\left( \mathbf{W},q\right) \right) \right)
=\sum_{G\in \mathcal{G}_{n}}p\left( G\right) \,\lambda \left(
P_{G}^{q}\right) ,
\end{equation}%
where $\lambda $ is the Lebesgue measure on $\Delta _{n}\times \left\{
q\right\} $. It is well--known that the volume of an intersection of an $n$%
--polytope $P$ and a moving hyperplane $\Theta _{t}$ sweeping $P$ over some
interval $\left( t_{0},t_{1}\right) \subset \mathbb{R}$ is a piecewise
polynomial function (spline) of $t$ of degree at most $n-1$ 
\citep[Theorem~3.2.1]{DeBoorHollig82,BieriNef83,Lawrence91,GritzmannKlee94}.
Thus, $\mathbb{E}\left( p\left( U\left( \mathbf{W},\cdot \right) \right)
\right) $ is a sum of splines of degree at most $n-1$, and accordingly also
a spline of the same or lower degree.
\end{proof}

\FloatBarrier%

\section{Concluding remarks}

In the present article we obtain a number of new analytical results,
including explicit formulae for the expected value and density of the voting
weight of the $k$--th largest player in a random weighted voting game, and
for the expected values of product--moments of voting weights, a
characteristic function of the distribution of the total weight of a random
coalition of players, and a general theorem about the functional form of the
relation between any quantity that is a function of the weighted
voting game and the quota. In addition, we note several regularities
appearing in numerical simulations that seem to provide promising subjects
for further study.

The results presented above enhance our understanding of the relationship
between voting game parameters, such as the Coleman efficiency index or
voting power, and the qualified majority quota $q$ in random voting games
where weights are drawn from the uniform distribution on the probability
simplex $\Delta _{n}$. These can have potential applications in the area of
voting rule design, especially if the rules are drafted behind a veil of
ignorance with regard to the actual distribution of players' weights (as is
the case for business corporations). Moreover, the results presented in
Sec.~2, regarding the distribution of voting weights of the $k$--th largest
player and the expected values of product--moments of voting weights, may
find applications in other areas of social choice theory. For instance,
Theorem \ref{thm:kPlayer} can be applied to obtain the probability of a
candidate with a specified vote share winning the election held under the
plurality rule.

Future work will focus on proving Conjecture \ref{con:critPoints};
developing a workable large--$n$ approximation on the basis of the normal
approximation of the Penrose--Banzhaf index; and generalizing the results
presented here for other Dirichlet measures.

\section*{Acknowledgments}

We wish to thank Jaros\l aw Flis for fruitful discussions and Geoffrey Grimmett for his insightful comments on an earlier draft of this article. We are privileged to acknowledge a long-term fruitful interaction with late Friz Haake, which triggered this work.


\section*{Appendix. Proof of Theorem \protect\ref{thm:kPlayer}}

Let $X_{1},\dots ,X_{n}\sim \func{Exp}\left( 1\right) $ be independent
random variables with densities $f_{X_{j}}\left( x\right) :=e^{-x}$ for
every $j=1,\dots ,n$ and $x>0$. As in the proof of Proposition \ref%
{propOrderStats}, we can assume that%
\begin{equation}
W_{k}^{\downarrow }=\frac{X_{k}^{\downarrow }}{\sum_{i=1}^{n}X_{i}}.
\end{equation}

By \citet[(2.1.3)]{DavidNagaraja03}, the order statistic $X_{k}^{\downarrow }$
($k=1,\dots ,n$) has an absolutely continuous distribution with the density given, for $x\in 
\mathbb{R}_{+}$, by%
\begin{equation}
f_{X_{k}^{\downarrow }}(x)=k\dbinom{n}{k}e^{-kx}\left( 1-e^{-x}\right)
^{n-k}.
\end{equation}%
Let $\Psi :=\sum\nolimits_{j=1}^{n}X_{j}$. By the Markov property of order
statistics \citep[Thm. 2.5]{DavidNagaraja03}, the conditional distribution of 
$X_{1}^{\downarrow },\dots ,X_{k-1}^{\downarrow }$ given $X_{k}^{\downarrow
}=y>0$, is the same as the distribution of order statistics $%
Y_{1}^{\downarrow },\dots ,Y_{k-1}^{\downarrow }$ from i.i.d.\ random variables $Y_{1},\dots
,Y_{k-1}$ with $Y_{j}\sim \func{Exp}\left( 1\right) $ truncated to $%
(y,\infty )$ for $j=1,\dots ,k-1$. Likewise, the conditional distribution of 
$X_{k+1}^{\downarrow },\dots ,X_{n}^{\downarrow }$ given $X_{k}^{\downarrow
}=y>0$, is identical to the distribution of order statistics $%
Z_{1}^{\downarrow },\dots ,Z_{n-k}^{\downarrow }$ from i.i.d.\ random variables $Z_{1},\dots
,Z_{n-k}$ such that $Z_{j}\sim \func{Exp}\left( 1\right) $ truncated to $%
(0,y)$ for $j=1,\dots ,n-k$. Moreover, we can choose $Y_{1},\dots ,Y_{k-1}$
and $Z_{1},\dots ,Z_{n-k}$ to be independent. Thus, for their sums we obtain
respectively:%
\begin{equation}
\left( \left( \sum_{j=1}^{k-1}X_{j}^{\downarrow }\right) \Bigg|%
X_{k}^{\downarrow }=y\right) \overset{d}{=}\sum_{j=1}^{k-1}Y_{j}^{\downarrow
}\overset{d}{=}\sum_{j=1}^{k-1}Y_{j},
\end{equation}%
i.e., the sum of $k-1$ independent exponential random variables variables
truncated to $(y,\infty )$, and%
\begin{equation}
\left( \left( \sum_{j=k+1}^{n}X_{j}^{\downarrow }\right) \Bigg|%
X_{k}^{\downarrow }=y\right) \overset{d}{=}\sum_{j=1}^{n-k}Z_{j}^{\downarrow
}\overset{d}{=}\sum_{j=1}^{n-k}Z_{j},
\end{equation}%
i.e., the sum of $n-k$ independent exponential random variables truncated to 
$(0,y)$. But it is easy to see that a sum of $k-1$ left--truncated
independent exponential random variables is a gamma--distributed random
variable with parameters $\left( k-1,1\right) $ shifted by a constant, $%
y\left( k-1\right) $. Thus,%
\begin{equation}\label{eq:psiCond}
\left( \Psi \bigg|X_{k}^{\downarrow }=y\right) =\left( (\Psi
-X_{k}^{\downarrow })\bigg|X_{k}^{\downarrow }=y\right) +y\overset{d}{=}\Xi ,
\end{equation}%
where $\Xi :=\sum_{j=1}^{k-1}Y_{j}+yk+\sum_{j=1}^{n-k}Z_{j}$, and $%
\sum_{j=1}^{k-1}Y_{j}\sim \limfunc{Gamma}\left( k-1,1\right) $ is
independent of $Z_{1},\dots ,Z_{n-k}$. Hence, the characteristic function of
their sum is given by the product of the characteristic functions:%
\begin{equation}
\varphi _{\sum_{j=1}^{k-1}Y_{j}}\left( t\right) :=\left( 1-it\right)
^{-(k-1)},
\end{equation}%
for $t\in \mathbb{R}$, and%
\begin{align}
\varphi _{Z_{j}}\left( t\right) &:=\frac{e^{y}}{1-e^{-y}}%
\int_{0}^{y}e^{itx-x}dx \notag \\
&=\left( 1-it\right) ^{k-n}\,\left( \frac{e^{y}-e^{ity}%
}{e^{y}-1}\right) ^{n-k},
\end{align}%
for $t\in \mathbb{R}$ and $j=1,\dots ,n-k$. Accordingly,%
\begin{equation}
\varphi _{\Xi }\left( t\right) =\left( 1-it\right) ^{1-n}\left( \frac{%
e^{y}-e^{ity}}{e^{y}-1}\right) ^{n-k}\,e^{ityk},
\end{equation}%
Applying the binomial theorem, we obtain%
\begin{multline}
\varphi _{\Xi }(t)=(e^{y}-1)^{k-n}\times \\ \sum\limits_{l=k}^{n}\dbinom{n-k}{l-k}%
\left( 1-it\right) ^{1-n}e^{y(n-l)}e^{i\pi (l-k)}e^{ityl}.
\end{multline}

As $\varphi _{\Xi }$ is integrable, for every $x\in \mathbb{R}_{+}$ we obtain by L\'{e}vy's inversion formula \citep[p.~347, (26.20)]{Billingsley95}:%
\begin{multline}\label{eq:densXi}
f_{\Xi }(x)=\frac{1}{2\pi }\int_{-\infty }^{+\infty }e^{-itx}\varphi _{\Xi
}(t)\,dt =\frac{1}{2\pi }(e^{y}-1)^{k-n}\times \\ \sum\limits_{l=k}^{n}\left( -1\right) ^{l-k}%
\dbinom{n-k}{l-k}e^{y(n-l)}\mathcal{\tciFourier }\left\{ \left( 1-it\right)
^{1-n}\right\} (x-yl).
\end{multline}%
By \citet[\S~3.2~(3), p.~118]{Bateman54}, we have%
\begin{equation}\label{eq:fourier}
\mathcal{\tciFourier }\left\{ \left( 1-it\right) ^{1-n}\right\} (s)=\left\{ 
\begin{array}{cc}
\frac{2\pi s^{n-2}e^{-s}}{\Gamma \left( n-1\right) }, & \qquad s>0, \\ 
0, & \qquad s\leq 0.%
\end{array}%
\right.
\end{equation}%
From (\ref{eq:psiCond}), (\ref{eq:densXi}), and (\ref{eq:fourier}),%
\begin{multline}
f_{\Psi |X_{k}^{\downarrow }\,=\,y}(x)=f_{\Xi }\left( x\right) =(e^{y}-1)^{k-n} \times \\ \sum\limits_{l=k}^{\min \left( n,\left\lfloor
x/y\right\rfloor \right) }\frac{\left( -1\right) ^{l-k}}{\left( n-2\right) !}%
\dbinom{n-k}{l-k}\left( x-yl\right) ^{n-2}e^{ny-x}.
\end{multline}

\noindent Thus, by \citet{Curtiss41}, the density of the ratio is given by 
\begin{gather}
f_{W_{k}^{\downarrow }}\left( x\right) =\int_{0}^{\infty }\left\vert
z\right\vert \,f_{X_{k}^{\downarrow },\Psi }(xz,z)\,dz \notag \\ = \int_{0}^{\infty
}\left\vert z\right\vert \,f_{\Psi |X_{k}^{\downarrow
}\,=\,xz}(z)\,f_{X_{k}^{\downarrow }}(xz)\,dz  \notag \\
=\int_{0}^{\infty }\left\vert z\right\vert \sum\limits_{l=k}^{T(n,x)}\frac{k\left( -1\right) ^{l-k}}{%
\left( n-2\right) !}\dbinom{n-k}{l-k}\dbinom{n}{k}\frac{\left( z-xzl\right)
^{n-2}}{e^{z}}\,dz  \notag \\
=\frac{k}{\left( n-2\right) !}\sum\limits_{l=k}^{T(n,x)}\Gamma \left( n\right)\dbinom{n-k}{l-k}\dbinom{n}{k}\frac{\left( -1\right) ^{l-k}}{\left( 1-lx\right) ^{2-n}}  \notag \\
=n\left( n-1\right) \dbinom{n-1}{k-1}\sum\limits_{l=k}^{T(n,x) }\dbinom{n-k}{l-k}\dbinom{n}{k}\frac{\left( -1\right) ^{l-k}}{\left( 1-lx\right) ^{2-n}},
\end{gather}%
where $T(n,x) := \min \left(
n,\left\lfloor 1/x\right\rfloor \right)$, as desired.

\bibliographystyle{elsarticle-harv}
\bibliography{power}

\end{document}